\documentclass[11pt]{article}       
\usepackage{geometry}               
\usepackage{graphicx}
\usepackage{caption}
\usepackage{subcaption}
\usepackage{enumerate}
\usepackage{hyperref}
\usepackage{upgreek}
\usepackage{todonotes}
\usepackage{cite}

\usepackage{amsmath} 
\usepackage{amssymb}  
\usepackage{amsthm}  
\usepackage{mathrsfs}
\usepackage{mathtools}
\usepackage{bm}
\usepackage{cite}
\usepackage{lipsum}
\usepackage{algpseudocode}
\usepackage{algorithm} 
\usepackage{color}
\usepackage{array}
\usepackage{multirow}
\newtheorem{proposition}{Proposition}

\newtheorem{lemma}{Lemma}
\newtheorem{theorem}{Theorem}

\newtheorem{cor}{Corollary}
\newcolumntype{M}[1]{>{\centering\arraybackslash}m{#1}}

\DeclareMathOperator*{\argmax}{arg\,max}

\newcommand{\R}{\mathbb R}
\newcommand{\kl}{D_{\mathrm{KL}}}
\newcommand{\ent}{\mathcal H}
\newcommand{\mmax}{M_{\text{max}}}  
  
\newcommand{\raa}{N_{\text{AA}}}
\newcommand{\naa}{n_{\text{AA}}}

\title{Anderson Acceleration for Partially Observable Markov Decision Processes: A Maximum Entropy Approach\thanks{This work was supported in part by  the National Research Foundation of Korea funded by the MSIT(2020R1C1C1009766), the Information and Communications Technology Planning and Evaluation (IITP) grant funded by MSIT(2020-0-00857), and Samsung Electronics. A preliminary version of this work was presented at the 60th IEEE Conference on Decision and Control~\cite{ermis2021on}.}}

\author{Mingyu Park
\and
Jaeuk Shin
\and
 Insoon Yang\thanks{M. Park, J. Shin, and I. Yang are with the Department of Electrical and Computer Engineering, Automation and Systems Research Institute,  Seoul National University, Seoul 08826, Korea, 
        {\tt\small \{pmg1202, sju5379, insoonyang\}@snu.ac.kr}} 
        }

\date{}

\begin{document}
\maketitle
\pagestyle{myheadings}
\thispagestyle{plain}

\begin{abstract}
Partially observable Markov decision processes (POMDPs) is a rich mathematical framework that embraces a large class of complex sequential decision-making problems under uncertainty with limited observations.
However, the complexity of POMDPs poses various computational challenges, motivating the need for an efficient algorithm that rapidly finds a good enough suboptimal solution. 
In this paper, 
we propose a novel accelerated offline POMDP algorithm exploiting Anderson acceleration (AA) that is capable of efficiently solving fixed-point problems using previous solution estimates.
Our algorithm is based on
the Q-function approximation (QMDP) method to alleviate the scalability issue inherent in POMDPs. 
Inspired by the quasi-Newton interpretation of AA, 
we propose a maximum entropy variant of QMDP,   which we call \emph{soft QMDP}, to fully benefit from AA. 
We prove that  the overall algorithm  converges to the suboptimal solution obtained by soft QMDP.
Our algorithm can also be implemented in a model-free manner using simulation data. 
Provable error bounds on the residual and the solution are provided to examine how the simulation errors are propagated through the proposed algorithm. 
 Finally, the performance of our algorithm is tested on several benchmark problems.   
 According to the results of our experiments, the proposed algorithm converges significantly faster without degrading the solution quality compared to its standard counterparts.     
\end{abstract}

\section{Introduction}\label{sec:intro}

Partially observable Markov decision processes (POMDPs) have served as a useful mathematical framework for
various sequential decision-making problems with limited noisy observations~\cite{Astrom1965}, including
autonomous driving~\cite{bai2015intention}, 
collision avoidance~\cite{holland2013optimizing}, repair problems~\cite{bhattacharya2020reinforcement}, 
and multi-agent systems~\cite{capitan2013decentralized}.
Unfortunately, computing an exact optimal policy of a POMDP is a challenging task~\cite{littman1996algorithms}.
When the system state is not fully observable, the agent at best maintains its belief over the possible system states. This makes the problem computationally intractable even in the case of finite POMDPs
as the policy is a function of the continuous belief instead of being a function defined over a discrete set of states.
To overcome this computational issue, many works have focused on finding an approximate solution via online or offline algorithms~\cite{kochenderfer2015decision}.
Online POMDP algorithms~\cite{Silver2010,ye2017despot} mostly use tree search in the belief space. However, the quality of the chosen actions heavily relies on the time budget of the tree search, which is problematic especially when the problem size is large.
On the other hand, the goal of offline POMDP solvers is to find an approximate value function determined by a finite set of vectors called \emph{$\alpha$-vectors}, relying on the fact that the value function of a POMDP is a piecewise linear and convex function~\cite{sondik1971optimal}.

Offline POMDP methods can be classified into two types depending on how the set of $\alpha$-vectors are approximated: 
$(i)$ point-based value iteration (PBVI) and its variants~\cite{pineau2003point,spaan2005perseus,kurniawati2008sarsop}, 
and $(ii)$ state-space methods~\cite{littman1995learning,hauskrecht2000value,ermis2021on}.
Algorithms in the first category retain a set of $\alpha$-vectors and add new $\alpha$-vectors to the set 
based on the point-based backup rule with sampled beliefs.
Despite their near-optimality,
those methods often suffer from memory complexity due to the exponential increase in the number of $\alpha$-vectors.
To resolve this issue, PBVI~\cite{pineau2003point} and PERSEUS~\cite{spaan2005perseus} algorithms explicitly bound the size of the $\alpha$-vectors at each iteration, and SARSOP~\cite{kurniawati2008sarsop} prunes the $\alpha$-vectors while exploring through the tree search in the belief space.
Still, these methods require a large set of $\alpha$-vectors and a higher per-step computational complexity compared to the algorithms in the second category.
Conversely, the state-space methods consider a fixed number of $\alpha$-vectors, thereby naturally circumventing the memory issue.
Specifically, the Q-function approximation, called QMDP~\cite{littman1995learning} and the fast informed bound (FIB)~\cite{hauskrecht2000value},  assign a single $\alpha$-vector to each action,
and thus the size of the $\alpha$-vector set remains fixed throughout all iterations. 
However, this approach requires a large number of iterations until convergence, and it is not straightforward to use the method when the exact model information is unavailable.

The main objective of this work is to develop an accelerated offline state-space POMDP algorithm to alleviate the issue of slow convergence.
Our method uses Anderson acceleration (AA) to speed up the solution of fixed-point problems~\cite{anderson1965iterative}.
Unlike the standard fixed-point iteration (FPI) that repeatedly applies an operator of interest to the most recent estimate, AA updates its iterate as a linear combination of previous iterates with weights obtained by solving a simple optimization problem.
Recently, there have been attempts to use AA in algorithms to solve Markov decision processes (MDPs) or reinforcement learning (RL) problems~\cite{geist2018anderson,shi2019regularized,sun2021damped,Ermis2020}.
It has been empirically shown  that AA improves convergence or learning speed in the fully observable setting.  
However, to the best of our knowledge, the use of AA in POMDPs has not yet been actively studied.

 Another key component of our method is the maximum entropy regularization of an approximate Bellman operator of interest. 
 The maximum entropy method has been successfully used in various MDP and RL algorithms by inducing a high-entropy policy that is known to encourage exploratory behaviors~\cite{rawlik2012stochastic,haarnoja2017reinforcement,haarnoja2018soft}.    
  While the maximum entropy method has further been generalized to a broader class of regularizers~\cite{geist2019theory}, a game-theoretic setting~\cite{firoozi2022exploratory}, and the regularization of occupancy measures~\cite{neu2017unified}, its use has been restricted to the fully observable MDPs. 
Notably, \cite{savas2022entropy} formulates the POMDP planning problem as an entropy maximization problem. However, this approach is fundamentally different from ours in that it
takes the entropy of the POMDP process as the objective function. 
It is also worth mentioning that other prior works~\cite{melo2005use,molloy2021active} use entropy in POMDP problems for different purposes; they focus on improving the control quality~\cite{melo2005use}
or reducing the uncertainty of the resulting trajectories~\cite{molloy2021active}.
Departing from these directions, 
we use the maximum entropy variant of the state-space method QMDP to rapidly compute an approximate solution of POMDP problems. 

In this paper, we propose a novel accelerated POMDP method with a convergence guarantee, exploiting both AA and maximum entropy regularization. 
For scalability, we adopt QMDP, which is a state-space algorithm that computes one $\alpha$-vector for each action. 
To speed up convergence using AA, we first  interpret QMDP as a solution to a fixed-point problem, and then propose  a maximum entropy variant of QMDP, which we call \emph{soft QMDP}, inspired by the quasi-Newton interpretation of AA. 
To effectively reject potentially bad AA candidates in our algorithm, an adaptive safeguarding technique is devised as well. 
The resulting algorithm converges to the unique fixed point of the soft QMDP operator and can be easily extended to the simulation-based setting without model information. 
Our empirical studies show that the proposed method
significantly outperforms the original QMDP and its AA variant on various benchmark problems in terms of the total number of iterations and the total computation time without degrading the quality of obtained policy. 

The main contributions of this work are summarized as follows:
\begin{enumerate}

\item We propose soft QMDP, which is the maximum entropy variant of QMDP, to fully benefit from AA and analyze how entropy regularization affects the fixed point (Section~\ref{sec:reg-pomdp}).

\item We design a novel POMDP algorithm that seamlessly combines AA and soft QMDP and prove its convergence to the soft QMDP solution (Section~\ref{sec:aa-pomdp}). 

 \item A simulation-based implementation of the proposed algorithm is also provided together with provable error bounds on the residual and the solution (Section~\ref{sec:sim}). 
  
\item According to the extensive experiments conducted on POMDP benchmark problems,  the proposed method significantly speeds up the convergence, saving up to $93\%$ of the number of iterations and  $52\%$ of the total computation time without degrading the quality of the computed policy  compared to the standard QMDP (Section~\ref{sec:num}). 
\end{enumerate}

\section{Background}

In this section, we provide essential notations and preliminaries for the proposed algorithm.

\subsection{Partially Observable Markov Decision Processes}

A partially observable Markov decision process (POMDP) is defined as a tuple $(S, A, T, R, Z, O, \gamma)$, where
\begin{itemize}

\item $S$ is the \emph{state space}, which is the set of all possible states of the environment;
\item $A$ is the \emph{action space}, which is the set of all possible actions executed by the agent;
\item $T: S \times A \times S \to [0,1]$ is the \emph{transition function} that describes the evolution of the environment over time. 
Specifically, given state $s \in S$ and action $a \in A$, $T (s, a, s') = \mathrm{Pr}(s' |s, a)$ denotes the probability of transitioning from the current state $s$ to the next state $s'$ when taking action $a$;
\item $R: S\times A \to [r_{\min}, r_{\max}]$ is the \emph{reward function} that quantifies the utility of each action for each state;
\item $Z$ is the \emph{observation space}, which is a finite set of all possible observations informed to the agent;
\item $O: S \times A \times Z \to [0,1]$ is the \emph{observation function}, 
where $O (s', a, z) = \mathrm{Pr}(z | a, s')$ denotes the probability of observing $z$ if action $a$ is executed and the resulting state is $s'$; and
\item $\gamma \in (0,1)$ is the \emph{discount factor}.
\end{itemize}
In this work, $S$ and $A$ are assumed to be finite sets. 

In a POMDP, events occur in the following order: 
In the current state $s_t$, the agent executes an action $a_t$, and the environment transitions to $s_{t+1}$ according to the transition function $T$. 
The agent then receives an observation $z_{t+1}$ related to $(a_t, s_{t+1})$ according to the observation function $O$. 
The objective of the agent is to maximize the expected discounted total reward, $\mathbb{E} [ \sum_{t=0}^\infty \gamma^t R(s_t, a_t)]$.

Note that the agent is unable to directly observe the states. Instead, the agent has access to observation $z_t \in Z$ that provides partial information about the state $s_t$. 
Let the \emph{history} at stage $t$ be defined as
$h_t := \{a_0, z_1, \ldots, z_{t-1}, a_{t-1}, z_t \}$, which contains all previous actions and observations. 
All relevant information from the history can be summarized in a \emph{belief state}~\cite{Astrom1965}, defined as
\[
b_t (s) := \mathrm{Pr}(s_t = s |h_t, b_0),
\]
which is the posterior distribution of being in state $s$ at stage $t$, given the complete history. 
It is well known that the belief state $b_t$ is a sufficient statistic for the history $h_t$~\cite{smallwood1973optimal}.
Given the previous action $a_{t-1}$ and the current observation $z_t$, the belief state can be updated via Bayes' theorem as follows:
\begin{equation}
\begin{split}
b_t(s') &= \pi(b_{t-1}, a_{t-1}, z_t)(s'):= \frac{1}{\mathrm{Pr}(z_t | b_{t-1}, a_{t-1})} O(s', a_{t-1}, z_t) \sum_{s \in S} T(s, a_{t-1},s') b_{t-1} (s),
\end{split}
\end{equation}
where $\pi (b, a, z)$ is called the \emph{belief state update function}.
The set of admissible belief states is given by $B := \mathcal{P} (S) = \{b \in \mathbb{R}^{|S|} : \sum_{s \in S} b(s) = 1, \; b(s) \geq 0 \: \forall s \in S \}$.

Since the belief state provides all information about the history needed for an optimal solution, a POMDP can be recast as an equivalent \emph{belief-state MDP}, a fully observable MDP whose state space is given by $B$.
In belief-state MDPs, a \emph{policy} $\mu : B \to \mathcal{P}(A)$ is defined as a mapping from the set of the belief states to the set of the probability distributions over $A$. Specifically, $\mu(a|b)$ denotes the probability of executing an action $a$ in a belief state $b$.
Then, an \emph{optimal policy} is an optimal solution to the following belief-state MDP problem:
\begin{align*}
& V^\star(b) :=  \max_{\mu \in M} \; \mathbb{E} \left[
\sum_{t=0}^\infty \gamma^t \sum_{s \in S} b_t (s) \sum_{a \in A} R(s, a) \mu (a|b_t) \middle\vert b_0 = b 
\right],
\end{align*}
where $M$ denotes the set of all admissible policies, and $V^\star: B \to \mathbb{R}$ is the optimal value function.

The solution of the belief-state MDP problem can be obtained via dynamic programming (DP). Specifically, 
$V^\star$ is the unique fixed point of the \emph{Bellman operator} $\mathcal{T}$, which is defined as
\begin{equation}\label{bellman}
\mathcal{T} V (b) := \max_{a \in A} \bigg [ R_B (b, a) + \gamma \sum_{z \in Z} \mathrm{Pr} (z | b, a) V(\pi (b, a, z)) \bigg ],
\end{equation}
where $R_B (b,a) := \sum_{s \in S} b(s) R(s, a)$ and 
$\mathrm{Pr} (z | b, a)$ can be computed as
$\sum_{s' \in S} O(s', a, z) \sum_{s \in S} T(s, a, s') b(s)$.
Therefore, $V^\star$ can
be computed by solving the \emph{Bellman equation} $V^\star = \mathcal{T} V^\star$. 

One of the most fundamental algorithms used to solve the Bellman equation is \emph{value iteration} (VI) which uses the contraction property of the Bellman operator. 
VI starts with an  initial guess $V_0$ and repeatedly performs the Bellman update $V_k := \mathcal{T} V_{k-1}$. 
The key observation made by Smallwood and Sondik~\cite{smallwood1973optimal} is that any iterate $V_k$ is piecewise linear and convex since the Bellman operator preserves the piecewise linearity and convexity of value functions. 
More precisely, $V_k$ can be represented by a set of $|S|$-dimensional vectors $\Gamma_k := \{\alpha_0^k, \ldots, \alpha_{|\Gamma_k|}^k \}$:
\[
V_k (b) = \max_{\alpha \in \Gamma_k} \sum_{s \in S} \alpha (s) b (s),
\]
where each $\alpha_i^k$ is called an $\alpha$-vector.
The standard VI algorithm makes an exact calculation on $\mathcal{T}$ by updating $\alpha$-vectors from the previous set $\Gamma_{k-1}$ to the current set $\Gamma_k$~\cite{sondik1971optimal,cassandra1971}.
Unfortunately, in the worst case, the size of this representation grows exponentially as $|\Gamma_k |$ is in $\mathcal{O}(|A| | \Gamma_{k-1} |^{|Z|})$, which makes the standard VI computationally intractable even for small problems.
Even performing a single iteration is known to be impractical~\cite{shani2013survey} since it requires $\mathcal{O}(|\Gamma_k| |A| |Z| + |A| |S| |\Gamma_k|^{|Z|})$ steps of computation.
Various approximation techniques have been proposed to tackle these issues.

\subsection{QMDP}
State-space POMDP algorithms approximate the value function with a small set of $\alpha$-vectors by associating each action with one $\alpha$-vector. 
Among them, the QMDP method is one of the most popular algorithms~\cite{littman1995learning}.
QMDP updates $\alpha$-vectors with the following rule:
\begin{equation}\label{eq:qmdp}
\alpha^{k+1}_a (s) := R(s, a) + \gamma \sum_{s'} T (s, a, s') \max_{a'} \alpha^k_{a'} (s').
\end{equation}
The update rule~\eqref{eq:qmdp} can be derived from the exact Bellman operator~\eqref{bellman} 
by first decoupling the expectation over $b(s)$ from $R_B(b, a)$,
and then decoupling expectation over $b(s)$ and $O(s',a,z)$ from $\mathrm{Pr}(z | b, a)$.

With an $\alpha$-vector obtained by QMDP, the policy is constructed as follows:
\begin{equation}\label{eq:state-space-policy}
\mu(a|b) = 
\begin{cases}
1 \quad \mbox{if} \quad a \in \argmax_{a \in {A}} b^\top  \alpha_a, \\
0 \quad \mbox{otherwise}.
\end{cases}
\end{equation}
Furthermore, the approximate value function is given by
\begin{equation}\label{eq:val}
\bar V (b) = \max_{a \in A} b^\top \alpha_a.
\end{equation}
During the iterations, the size of $\Gamma_k := \{ \alpha_1^k, \ldots, \alpha_{|A|}^k\}$ remains the same, thereby alleviating the scalability issue of VI in POMDPs.
Precisely, the complexity of each iteration in QMDP is in $\mathcal{O}(|A| |S|^2)$, which is significantly lower than that of the original VI.
To enjoy the benefit, we employ the QMDP method instead of seeking an exact solution.

\subsection{Anderson Acceleration}\label{sec:aa}

In this work, we consider QMDP as an iterative algorithm to find the fixed point of a nonlinear operator that approximates the original Bellman operator. 
Indeed, fixed-point problems arise from various domains ranging from sequential decision making to game theory~\cite{zhang2021stabilizing}. Formally, a fixed-point problem associated with a function $F: \mathbb{R}^N \to \mathbb{R}^N$ aims to find $x \in \mathbb{R}^N$ such that
\[
x = F(x).
\]
As an algorithm to solve such problems, a \emph{fixed-point iteration} (FPI) algorithm repetitively applies the function $F$ to the latest iterate to obtain a better estimate of the fixed point: 
\[
x^{k+1} := F(x^k).
\]
It follows from the Banach fixed-point theorem that the FPI algorithm converges to the unique fixed point of $F$ when the function $F$ is a contraction mapping~\cite{puterman2014markov}. 

To speed up the process of finding the fixed point, \emph{Anderson acceleration} (AA)~\cite{anderson1965iterative} can be used instead of the naive FPI algorithm.
Unlike FPI, AA uses information about the past iterates.
Specifically, at the $k$th iteration, AA maintains the most recent $M^k +1$ estimates $(x^k, \ldots, x^{k- M^k} )$ in memory, where $M^k := \min\{M, k\}$ is the memory size and $M$ is a hyper-parameter~\cite{walker2011anderson}.
The estimate of $x$ is then updated as the following weighted sum of 
$F (x^{k-M^k + i} )$:
\begin{equation}\label{eq:aa_update}
x^{k+1} := \sum_{i=0}^{M^k} w_i^k F  (x^{k- M^k + i} ).
\end{equation}
Here, the weight vector $w^k = \left(w_0^k, \ldots, w_{M^k}^k \right) \in \mathbb{R}^{M^k +1}$ is an optimal solution of the following quadratic programming (QP) problem:
\begin{equation}\label{eq:aa_problem}
\begin{split}
\min \quad & \bigg \| \sum_{i=0}^{M^k} w_i^k G\left(x^{k-M^k + i} \right) \bigg \|_2^2\\
\mbox{s.t.} \quad & \sum_{i=0}^{M^k} w_i^k =1,
\end{split}
\end{equation}
where
\[
G(x) := x- F(x)
\]
is the residual function. 
Accordingly, the memory is updated with the most recent $M^{k+1}+1$ iterates $(x^{k+1}, \ldots, x^{k+1 - M^{k+1}})$ before repeating the process. 
The update rule \eqref{eq:aa_update} can be interpreted as an extrapolation mechanism that exploits the previous iterates in the memory to
rapidly reduce the residual.

As introduced in recent literature~\cite{fang2009two,toth2015convergence,evans2020proof}, the QP problem~\eqref{eq:aa_problem} can be reformulated through the following change of variables:
\begin{equation}\label{eq:relation}
w_i^k = \begin{cases} \xi_0^k & \mbox{if }i = 0 \\
\xi_i^k - \xi_{i-1}^k & \mbox{if }0 < i < M^k \\
1 - \xi_{M^k -1}^k & \mbox{if }i = M^k.
\end{cases}
\end{equation}
This is simply an affine transformation from the $(M^k - 1)$-dimensional affine subspace $\{w \in \mathbb{R}^{M^k} : \textbf{1}^\top w = 1 \}$ onto $\mathbb{R}^{M^k - 1}$.
We further let
\begin{equation}\label{eq:gys}
g^k := G (x^k), \quad y^k := g^{k+1} - g^k, \quad s^k := x^{k+1} - x^k,
\end{equation}
and
\begin{equation}\label{eq:ys}
Y_k := \begin{bmatrix}
y^{k - M^k} & \cdots & y^{k-1}
\end{bmatrix}, \:
S_k := \begin{bmatrix}
s^{k - M^k} & \cdots & s^{k-1}
\end{bmatrix}.
\end{equation}
Then, the original problem~\eqref{eq:aa_problem} is transformed into the following unconstrained least-squares problem
\begin{equation}\label{eq:aa-problem}
\min_{\xi} \left\| g^k - Y_k \xi \right\|^2_2.
\end{equation}
Using~\eqref{eq:relation}, the optimal solution $w^k$ of~\eqref{eq:aa_problem} is easily constructed from the optimal solution $\xi^k$ of~\eqref{eq:aa-problem}.
Assuming that $Y_k$ has a full column rank, the problem \eqref{eq:aa-problem} admits the following closed-form solution: 
\[
\xi^k = (Y_k^\top Y_k)^{-1} Y_k^\top g^k.
\]
However, the problem~\eqref{eq:aa-problem} is prone to numerical instability~\cite{scieur2016regularized} as the condition number of $Y_k^\top Y_k$ can be large depending on the values of the estimates stored in the memory.
This issue is often alleviated using additional regularization techniques~\cite{scieur2016regularized,shi2019regularized,fu2019anderson}.
Among them, we use the Tikhonov regularization in AA proposed in~\cite{fu2019anderson}.

\section{Entropy Regularization in POMDPs}\label{sec:reg-pomdp}

As the first step toward designing an efficient AA algorithm for POMDPs, 
we interpret QMDP as an algorithm to find the fixed point of a nonlinear operator and propose a maximum entropy variant of the operator motivated by the interpretation of AA as the quasi-Newton method. 
The key intuition is that applying  entropy regularization flattens the landscape of the operator,
as explained in Section~\ref{subsec:reg-intuition}.

\subsection{Soft QMDP}

We first propose the maximum entropy variant of QMDP, which we call the soft QMDP. 
To begin, let $\alpha$ denote the vector representation of the set of $\alpha$-vectors $\Gamma = \{ \alpha_1, \ldots, \alpha_{|A|}\}$ in QMDP:
\[
\alpha := \begin{bmatrix}
\alpha_1^\top & \alpha_2^\top & \cdots & \alpha_{|A|}^\top \end{bmatrix}^\top \in \R^{N},
\]
where $N := |S| |A|$.
We define the QMDP operator $F: \R^{N} \to \R^{N}$ as
\begin{equation}\label{eq:qmdp_op}
\begin{split}
(F \alpha) ( s, a) := R(s,a) + \gamma \sum_{s' \in S} T(s, a, s') \max_{a' \in A} \alpha (s', a').
\end{split}
\end{equation}
Then, the QMDP update rule~\eqref{eq:qmdp} can be compactly expressed as
\[
\alpha^{k+1} = F \alpha^k.
\]
We can show that $F$ is a contraction mapping, and thus QMDP converges to the unique fixed point of $F$.
\begin{lemma}\label{lem:contraction-qmdp}
The QMDP operator $F$ is a $\gamma$-contraction mapping with respect to $\| \cdot \|_\infty$. That is, for all $\alpha, \alpha' \in \R^{N}$, we have
\[
\| F \alpha - F \alpha' \|_\infty  \leq \gamma \| \alpha - \alpha' \|_\infty.
\]
\end{lemma}
Its proof can be found in Appendix~\ref{app:lem1}.

We now rewrite the QMDP operator as
\begin{align}\label{eq:qmdp-phi-vector}
\begin{split}
& (F \alpha) ( s, a) =   R(s,a) + \gamma \sum_{s' \in S} T(s, a, s') \max_{\phi \in \mathcal{P}(A)}  \sum_{a'} \phi(a') \alpha (s', a'),
\end{split}
\end{align}
where the optimization variable $\phi$ represents a probability distribution over the action space, which we simply call the \emph{action distribution}. 
Adding the entropy $\mathcal{H} (\phi)$\footnote{The entropy of $\phi$ is defined as $\mathcal{H}(\phi) := -\sum_{a} \phi(a) \ln \phi(a)$.}
 of action distribution $\phi$ as a regularizer to the above maximization problem, we have
\begin{align}
\ent^*_{\tau}(\alpha(s, \cdot)) &:= \max_{\phi \in \mathcal{P}(A)}\left[\sum_{a}\phi(a)\alpha(s, a) + \tau \ent(\phi) \right] \label{eq:conjugate}\\
&= \tau \ln \left ( \sum_a \exp \frac{\alpha(s, a)}{\tau} \right ), \quad \alpha \in \mathbb{R}^{N},  \label{eq:conjugate2}
\end{align}
where $\tau$ is called a \emph{temperature parameter} that controls the level of the regularization.\footnote{In addition to the maximum entropy regularization, a variant of using the KL divergence is also applicable. Also, the FIB method may serve as a substitute for QMDP. 
For those variants, we refer the reader to Appendix~\ref{appx:reg-fib} and~\ref{appx:mellowmax}. Their empirical analyses can be found in Appendix~\ref{appx:exp} as well.} 
Note that $\mathcal{H}_\tau^*$ is the convex conjugate of $\tau\ent$.
Given $s$, let $\phi_s^\star$ be an optimal solution to \eqref{eq:conjugate}. 
Then, the \emph{maximum entropy policy} can be constructed as
 $\mu (\cdot | s) = \phi_s^\star$.
 While such a policy cannot be used due to partial observability, it provides a connection to  maximum entropy stochastic control in the fully observable setting~\cite{haarnoja2017reinforcement,geist2019theory,Kim2022}.

Using $\mathcal{H}_\tau^*$, we introduce the following modified version of the QMDP operator:
\begin{equation}\label{eq:qmdp-maxent}
\begin{split}
&(F_\tau \alpha) (s, a)  = R(s,a) + \gamma \sum_{s' \in S} T(s, a, s')  \tau \ln \left ( \sum_{a' \in A} \exp \frac{\alpha(s', a')}{\tau} \right ),
\end{split}
\end{equation}
which we call the \emph{soft QMDP} operator since \eqref{eq:conjugate} or \eqref{eq:conjugate2} can be interpreted as the soft maximum of its standard counterpart $\max_{a' \in A} \alpha (s', a')$
in~\eqref{eq:qmdp_op}.
This type of  entropy regularization is known to preserve the contraction property of the operator~\cite{geist2019theory}. 
The corresponding soft QMDP algorithm is given by
$\alpha^{k+1} = F_\tau \alpha^k$ that converges to the unique fixed point of $F_\tau$.

The difference between the soft QMDP solution and the original one can be bounded using the classical analysis for error bounds in approximate DP~\cite{bertsekas1996neuro}.
\begin{proposition}\label{pro:alpha-bound}
Let $\alpha^\star$ and $\alpha^\star_\tau$ denote 
 the fixed points of $F$ and $F_\tau$, respectively.
Then, $\alpha^\star(s, a) \leq \alpha^\star_\tau (s, a)$ for all $(s, a) \in S \times A$, and the difference between these fixed points is bounded as follows:
\begin{equation}
\|\alpha^\star - \alpha^\star_\tau\|_\infty \leq \frac{\gamma\tau\ln|A|}{1-\gamma}.
\end{equation}
\end{proposition}
Its proof can be found in Appendix~\ref{app:prop}.
Recall that the approximate value function $\bar{V}$ obtained by QMDP is given by \eqref{eq:val}. Similarly, the approximate value function $\bar{V}_\tau$ for soft QMDP can be computed as $\bar{V}_\tau(b) = \max_{a \in A} b^\top \alpha_{\tau, a}^{\star}$. 
It directly follows from Proposition~\ref{pro:alpha-bound} that the gap between $\bar{V}$ and $\bar{V}_\tau$ admits the same bound, i.e., 
$\bar{V}(b) \leq \bar{V}_\tau(b) \leq \bar{V}(b) + \frac{\gamma\tau \ln|A|}{1-\gamma}$ for all $b \in B$.

\subsection{Why Does Maximum Entropy Regularization Help Anderson Acceleration?}\label{subsec:reg-intuition}

We now provide an informal but intuitive explanation of why applying the proposed entropy regularization is beneficial for AA.
To do so, we first illustrate the connection between AA and the quasi-Newton method,
akin to the previous discussions in~\cite{fang2009two,wei2021stochastic,grand2021convex}.\footnote{The relationship between AA and the quasi-Newton method is utilized in~\cite{brezinski2020shanks} in proving the
local linear convergence of AA algorithms. 
Furthermore, another work~\cite{wei2021stochastic}  exploits this connection to apply AA for stochastic gradient descent.
In addition, the analysis of AA or FPI from the optimization viewpoint is an active topic of research which has been recently employed to study accelerated variants of the value iteration~\cite{goyal2022first}.}
Suppose there exists a function $H:\mathbb{R}^N \to \mathbb{R}^N$ that satisfies
\[
\nabla H(x) = G(x) = x - F(x).
\]
Then, the fixed point $x^\star$ of $F$ is a critical point of $H$ as it satisfies $\nabla H(x^\star) = 0$,
and the single FPI with $F$ corresponds to a one-step gradient descent on $H$ with stepsize 1, i.e.,
\[
x^{k+1} = F\left(x^k \right) = x^k - \nabla H\left(x^k \right).
\]
Extending this point of view, AA can be interpreted as the quasi-Newton method for minimizing $H$.
To see this, we assume that $H$ is twice continuously differentiable and the estimates in the memory 
$(x^{k-1}, \ldots, x^{k-M^k})$ are in the small neighborhood $U_{x^k}$ of the current estimate $x^k$. Since $\nabla^2 H(x)$ is close to $\nabla^2 H \left(x^k \right)$ within $U_{x^k}$, Taylor's theorem can be used to obtain
\begin{align*}
\nabla^2 H\left(x^k \right) s^\ell &= \nabla^2 H\left(x^k \right) \left(x^{\ell + 1} - x^{\ell} \right) \\
 &\approx \nabla H\left(x^{\ell+1}\right) - \nabla H\left(x^{\ell}\right) = g^{\ell + 1} - g^\ell = y^\ell
\end{align*}
for all $k - M^k \leq \ell < k$. This implies $Y_k \approx  H\left(x^k \right) S_k$, which motivates the following approximate version of the original optimization problem~\eqref{eq:aa-problem}:
\begin{equation*}\label{eq:approx_qp}
\min_{\xi^k}\left\| \nabla H\left(x^k \right) - \nabla^2 H\left(x^k \right) S_k \xi^k \right\|^2_2.
\end{equation*}
Therefore, when $\xi^k$ is a solution of \eqref{eq:aa-problem} and $\nabla^2 H \left(x^k \right)$ is  invertible, $\nabla^2 H\left(x^k \right)^{-1} \nabla H\left(x^k \right)$ is a reasonable estimate of $S_k \xi^k$.
Finally, we arrive at the following approximation that provides a concrete interpretation of AA as the quasi-Newton method:
\begin{equation*}
\begin{split}
x^{k+1} &= \sum^{M^k}_{i=0} w^k_i x^{k-M^k+i} - \sum^{M^k}_{i=0} w^k_i g^{k-M^k+i} \\
&= \left(x^k - S_k\xi^k \right) - \left(g^k - Y_k\xi^k \right) \\
&\approx x^k - \nabla^2 H\left(x^k \right)^{-1} \nabla H\left(x^k\right),
\end{split}
\end{equation*}
where the term $g^k - Y_k \xi^k$ is neglected in the last equality as $\xi^k$ is the optimal solution of \eqref{eq:aa-problem}.

To see why our maximum entropy regularization of POMDPs speeds up the convergence of AA, 
we point out that increasing $\tau$ promotes randomness caused by the action distribution $\phi$ in~\eqref{eq:conjugate}.
Whenever the temperature parameter $\tau$ is large, the maximizer $\phi^\star$ of~\eqref{eq:conjugate} becomes independent of $\alpha$, and thus $\phi^\star$ tends to the uniform distribution over $A$. The resulting $F_\tau$ is nearly constant on some bounded domain. 
The same conclusion is deduced from  the following inequality which is a direct consequence of Jensen's inequality:
\begin{equation}\label{eq:entropy-ineq}
\frac{1}{|A|}\sum_{a} \alpha(s, a) + \tau \ln|A| 
\leq \mathcal{H}^\star_\tau(\alpha(s, \cdot)) \leq \max_a \alpha(s, a) + \tau \ln |A|.
\end{equation}
In  words, the asymptotic behavior of $\mathcal{H}^\star_\tau$ with respect to $\tau$ is dominated linearly by $\tau$, while the effect of $\alpha$ becomes negligible.
Therefore, on any compact subset of $\mathbb{R}^{N}$, $F_\tau$ can be regarded as a constant function if $\tau$ is sufficiently large.
In particular,
$\nabla H(\alpha) = \alpha - F_\tau (\alpha)$ becomes close to an affine function of $\alpha$, and $H$ is well-approximated by a quadratic function of $\alpha$. Therefore, $\nabla^2 H$ has a very small Lipschitz constant $L$ (at least locally). 
Since Newton's method enters the quadratically convergent stage within a small number of iterations when $L$ is sufficiently small, AA is expected to converge quickly when a sufficiently large $\tau$ is used.
This motivates us to use the proposed maximum entropy variant of POMDPs or, more specifically, soft QMDP in designing an efficient AA algorithm for solving POMDPs.  
A rigorous convergence analysis of the proposed method is provided in Section~\ref{sec:conv}.

\section{Anderson Acceleration for Approximate POMDP Solutions}\label{sec:aa-pomdp}

Probably the most straightforward way to design an AA algorithm for solving POMDPs is to employ the AA update rule~\eqref{eq:aa_update}  for computing the fixed point of the QMDP operator. 
However, this naive approach has a critical issue. 
As mentioned in Section~\ref{sec:aa}, the standard AA is often subject to numerical instability as well as a lack of global convergence guarantees except for the restricted class of fixed-point operators such as linear contraction operators~\cite{toth2015convergence}. 
To resolve this issue, we use a stabilized variant of AA proposed in \cite{fu2019anderson} together with a novel safeguarding mechanism.

\subsection{Stabilized AA for Soft QMDP}

We first introduce the stabilized variant of AA~\cite{scieur2016regularized}, which we call the stabilized AA (SAA), with application to the soft QMDP operator $F_\tau$. 
The residual function is given by
$G(\alpha) := \alpha - F_\tau \alpha$.
Then, $g^k, y^k$, and $s^k$ are defined in a similar manner as~\eqref{eq:gys} and $Y_k$ and $S_k$ are defined as \eqref{eq:ys}.
In the $k$th step, the $\alpha$-vector is updated as
\begin{equation}\label{eq:aapomdp-update}
\alpha^{k+1} := \sum_{i=0}^{M^k} w_i^k F_\tau \alpha^{k- M^k + i},
\end{equation}
where the weights $w^k$ are obtained by solving the following QP problem:
\begin{equation}\label{eq:raa-problem}
\min_{\xi} \; \| g^k - Y_k  \xi \\|_2^2 + 
\eta_k \| \xi^k \|_2^2,
\end{equation}
where $\eta_k = \eta ( \| S_k \|_F^2 + \|Y_k \|_F^2 )$ is a regularization coefficient with a hyperparameter $\eta > 0$.
 This is a Tikhonov regularized variant of~\eqref{eq:aa-problem} proposed in~\cite{fu2019anderson} and adopted in recent works~\cite{ermis2021on,sun2021damped}. While this type of adaptive regularization resolves the numerical instability of the original problem for AA~\eqref{eq:aa-problem}, the effect of regularization vanishes as the iterates get close to the fixed point and the size of $S_k$ and $Y_k$ tend to be small. Note that the analytic solution of~\eqref{eq:raa-problem}, denoted by $\xi^k$, is readily obtained as
\begin{equation}\label{eq:sol-xi}
\xi^k = \left(Y_k^\top Y_k + \eta_k I \right)^{-1} Y_k^\top g^k,
\end{equation}
where the extra term $\eta_k I$ prevents 
$\left(Y_k^\top Y_k + \eta_k I \right)$ from being near-singular. Furthermore, both unregularized AA and FPI can be regarded as extreme cases of~\eqref{eq:raa-problem}. Indeed, the unregularized version of AA is obtained by taking $\eta \rightarrow 0$, while taking $\eta \to \infty$ leads to $\xi^k \to 0$ which corresponds to FPI.

Using the solution of~\eqref{eq:sol-xi}, the AA update rule~\eqref{eq:aapomdp-update} can be written in the following matrix form:
\begin{equation*}
\begin{split}
\alpha^{k+1} &= (\alpha^k - S_k\xi^k) - (g^k - Y_k\xi^k) \\
&= \alpha^k - \left (g^k + (S_k - Y_k)\xi^k \right ) = \alpha^k - A_k g^k,
\end{split}
\end{equation*}
where $A_k$ is given as
\[
A_k = -I - (S_k - Y_k)\left(Y_k^\top Y_k + \eta_k I \right)^{-1} Y_k^\top.
\]
Then, a simple argument (e.g. see~\cite{fu2019anderson}) shows that the size of $A_k$ is bounded  as follows:
\begin{equation}\label{eq:aa-mat-bound}
\|A_k\|_2 \leq 1 + 2\eta^{-1}.
\end{equation}
Therefore,  Tikhonov regularization prevents the estimates from fluctuating severely, thereby alleviating the instability issue of AA.

Given the solution $\xi^k$ of~\eqref{eq:raa-problem}, an \emph{acceleration factor} $\theta^k$ is defined as:
\begin{equation}\label{eq:opt-gain}
\theta^k := \frac{\|g^k_w\|_2}{\|g^k\|_2}, \quad g^k_w := g^k - Y_k  \xi^k,
\end{equation}
where $g^k_w$ is the weighted sum of the residuals with weight vector $w^k$ obtained from the optimal solution $\xi^k$ of~\eqref{eq:raa-problem}. This is because using~\eqref{eq:relation}, $g^k - Y_k \xi^k$ can be expressed as
\begin{align}\label{eq:aa-equiv}
\begin{split}
g^k - Y_k \xi^k &= g^k - \sum_{\ell = 0}^{M^k-1} y^{k - M^k+ \ell} \xi^k_{\ell} \\
&= g^k - \sum_{\ell = 0}^{M^k-1} (g^{k - M^k+\ell+1} - g^{k-M^k +\ell}) \xi^k_{\ell} \\
&= (1 - \xi^k_{M^k-1})g^k + \sum_{i=1}^{M^k-1}(\xi^k_{i-1} - \xi^k_{i}) g^{k - M^k+i} + \xi^k_0 g^{k-M^k} \\
&= w^k_{M^k} g^k + \sum_{i=1}^{M^k-1} w^k_i g^{k-M^k+i}  + w^k_0 g^{k-M^k} \\
&= \sum_{i=0}^{M^k} w^k_i g^{k-M^k + i},
\end{split}
\end{align}
where the square norm of the final expression is taken as the objective function of~\eqref{eq:aa_problem}. Note that taking $\xi = \bm{0}$ results in $\|g^k_w\|_2 \leq \|g^k - Y_k \bm{0}\|_2 = \|g^k\|_2$, which implies that $\theta^k \leq 1$. The acceleration factor, introduced in~\cite{evans2020proof}, has been used to describe the decay rate of the residuals in AA steps, thus quantifing the advantage of AA over FPI~\cite{pollock2021anderson,sun2021damped,ouyang2020nonmonotone}.

\subsection{Algorithm}\label{subsec:safeguarding}

The notion of the target acceleration factor motivates a simple yet effective strategy to improve the performance of SAA. Roughly speaking, this strategy directly compares the  candidate obtained by SAA and the standard FPI candidate and chooses the better one by employing the technique known as \emph{safeguarding}~\cite{fu2019anderson,ouyang2020nonmonotone}. Specifically, we introduce the notion of a \emph{target acceleration factor} and accordingly design a novel safeguarding mechanism.
Our new algorithm, which combines the new and conventional safeguarding rules, converges to the fixed point of the soft QMDP operator $\alpha_\tau^\star$.

\subsubsection{Conventional Safeguarding}
 
 The conventional safeguarding mechanism checks the residual norm $\|g^k\|_\infty$ at each iteration and decides when to reject the AA candidate and switch to the FPI candidate.
For instance, in~\cite{fu2019anderson,zhang2021stabilizing}, safeguarding at the $k$th step checks if
\begin{equation}\label{eq:loose-safeguard}
\|g^k\|_\infty > C(\naa(k) + 1)^{-(1+\phi)},
\end{equation}
where $\naa(k)$ denotes the number of AA steps invoked until the $k$th step, and $C$ and $\phi$ are positive constants. We refer to the right-hand side of~\eqref{eq:loose-safeguard} as the \emph{target residual}.
The condition~\eqref{eq:loose-safeguard} means that the AA candidate is rejected if the residual norm $\|g^k \|_\infty$ is greater than the target residual.

However, it is worth noting that the safeguarding criterion~\eqref{eq:loose-safeguard} rarely rejects AA candidates, as the term $(\naa(k) + 1)^{-(1+\phi)}$ vanishes with a polynomial rate as $\naa(k) \to \infty$. Therefore, for a sufficiently large $k$, AA candidates are rarely rejected. 
Some recent literature including~\cite{henderson2019damped} has introduced more strict conditions, namely  using
$\gamma\|g^{k-1}\|_\infty$ as the target residual. However, this is often too conservative to fully exploit the advantage of AA.

\subsubsection{Target Acceleration Factor and New Safeguarding}

We now 
introduce a new notion called the \emph{target acceleration factor} and
propose a novel safeguarding mechanism.
The target acceleration factor is devised to  adaptively switch between the AA and FPI candidates in each iteration.
Specifically, we define the \emph{target acceleration factor} as follows:
\begin{equation}\label{eq:targ-gain}
\theta_\mathrm{targ} := \bar{m} - m \|g^k_w\|^2_2,
\end{equation}
where $m> 0$ and $0 < \bar{m} \leq 1$ are predefined hyperparameters.

Our new safeguarding mechanism checks if $\theta^k  \geq \theta_\mathrm{targ}$ in which case  the AA candidate is rejected in favor of the FPI candidate. 
Note that the rejection frequently occurs 
 when $k$ is small: during the early stages of the algorithm, $\Vert g^k_w \Vert_2$ tends to be large, which leads to a small value of $\theta_{\mathrm{targ}}$. 
Employing the target acceleration factor avoids using the  solution of~\eqref{eq:raa-problem} in early stages, where the solution quality is usually poor due to the potentially large distance between the initial estimate and the fixed point. 
 Furthermore, the choice of $g^k_w$ instead of $g^k$ in~\eqref{eq:targ-gain} prevents the condition from being overly conservative since $\Vert g^k_w\Vert_2 \ll \Vert g^k \Vert_2$ in practice. 

Unlike the target residual that only considers the number of AA iterations taken so far, the target acceleration factor~\eqref{eq:targ-gain} is adaptive in that it actively exploits the residuals of the recent iterates $g^k, \ldots, g^{k-M^k}$. Therefore, it is not surprising that the use of the target acceleration factor provides better chances of rejecting bad AA candidates compared to the target residual.
Unfortunately,
the corresponding safeguarding rule does not directly guarantee the convergence of the sequence $\{\alpha^k\}$.
As a remedy, we combine the new safeguarding mechanism with the conventional one~\eqref{eq:loose-safeguard}.
We refer to this combination as the \emph{double safeguarding}.

\begin{algorithm}[t!]
\caption{AA-sQMDP}\label{alg:aa-qmdp}
\begin{algorithmic}[1]
	\State \textbf{Input}: Initial vector $\alpha^{0}$, soft QMDP operator $F_\tau$, memory size $\mmax$, regularization parameter $\eta$, 
	safeguarding constants $D > 0$, $\phi > 0$, $N_s \geq 1$, $m > 0$, $\bar{m} > 0$;
	\State Initialize $I :=$ True, $\naa := 0$, $\raa := 0$;
	\State Compute  $\alpha^{1} := F \alpha^{0}$, $g^{0} := G (x^0)$;
	\For{$k = 1,2, \dots $}
	\State Set $M^k := \min\{M, k\}$;
	\State Compute the FPI candidate: $\alpha_{\text{FPI}} := F_\tau \alpha^{k}$;
	\State Compute the residual: $g^{k} := G (\alpha^k)$;
	\State \text{\# \bf Stabilized Anderson acceleration}
	\State Compute $w^k$ from~\eqref{eq:sol-xi} and~\eqref{eq:relation}
	\State Compute AA candidate $\alpha_{\text{AA}}$ from~\eqref{eq:aapomdp-update};
	\State Compute $\theta^k$ from $w^k$ and $g^k$;
	\State \text{\# \bf Double safeguarding}
	\If{$\theta^k > \Bar{m} - m \|g^k_w\|^\kappa_2$}
	\State $\alpha^{k+1} := \alpha_{\text{FPI}}$, $\raa := 0$;
	\ElsIf{$I = \text{True}$ or $\raa \geq N_s$}
	\If{$\|g^k\|_\infty \leq D\|g^{0}\|_\infty(\frac{\naa}{N_s}+1)^{-(1+\phi)}$}
	\State $\alpha^{k+1} := \alpha_{\text{AA}}$;
	\State $\naa := \naa+1$, $\raa := 1$, $I := \text{False}$;
	\Else
	\State $\alpha^{k+1} := \alpha_{\text{FPI}}$, $\raa := 0$;
	\EndIf
	\Else
	\State $\alpha^{k+1} := \alpha_{\text{AA}}$;
	\State $\naa := \naa+1$, $\raa := \raa + 1$; 
	\EndIf
	\State Terminate if the stopping criterion is satisfied;
	\EndFor
	\State \textbf{return} $\alpha^{k+1}$;
\end{algorithmic}
\end{algorithm}

\subsubsection{AA-sQMDP Algorithm and Its Convergence}\label{sec:conv}

Evaluating the acceleration factor of SAA enables us to adaptively compare the quality of the AA candidate with that of the FPI candidate.
This motivates the AA algorithm for POMDPs, which we call \emph{AA-sQMDP}, 
 equipped with the double safeguarding mechanism.
The resulting algorithm  is described in Algorithm~\ref{alg:aa-qmdp}.
The algorithm starts with $\alpha^0$, with each element randomly chosen from $[{r_{\min}}/{(1-\gamma)},{r_{\max}}/{(1-\gamma)}]$. 
At each step, both of the FPI  and AA candidates are computed (lines 6--10). 
Then, our double safeguarding criterion  is used to decide whether to select the AA candidate instead of the FPI candidate (lines 12--25). Specifically, the algorithm first examines the acceleration factor $\theta^k$ and rejects the AA candidate if $\theta^k$ is greater than the target acceleration factor (line 13). Otherwise, the algorithm continues to check whether the residual norm $\| g^k \|_\infty$ is sufficiently small compared to the target residual (line 16). The target residual test  is omitted for $N_s$ steps after accepting the AA candidate.

The convergence of AA-sQMDP to the fixed point $\alpha^\star$ of the soft QMDP operator is shown as follows. 

\begin{theorem}\label{thm:global-conv}
Let $\{\alpha^k\}$ be the sequence generated by Algorithm~\ref{alg:aa-qmdp}.
Then, $\{\alpha^k\}$ converges to the fixed point $\alpha^\star_\tau$ of the soft QMDP operator $F_\tau$, i.e.,
 \[
 \lim_{k\rightarrow\infty} \alpha^k = \alpha^\star_\tau.
 \]
\end{theorem}

 \begin{proof}
Suppose first that AA candidates are selected for finitely many times. In this case, it is straightforward to show the convergence of $\alpha^k$ to $\alpha^\star_\tau$; There exists $K \geq 0$ such that FPI is used exclusively for all $k \geq K$. Then, we have
\begin{equation*}
\lim_{k\rightarrow\infty} \alpha^k = \lim_{k\rightarrow\infty}F_\tau^{k - K}\alpha^{K} = \alpha^\star_\tau.
\end{equation*}

We now assume that AA candidates are selected for infinitely many times. Let $k_i$ denote the initial iteration count for accepting an AA candidate.\footnote{If there exists $i$ such that AA is selected for all $k \in [k_i, k_i+N_s)$, then we simply set $k_{i+1} = k_i + N_s$ so that $k_i$ is defined for all $i \geq 0$.} Then, there exists an integer $\ell_i \in (k_i, k_{i+1}]$  such that AA is used over $[k_i, \ell_i)$ and FPI is used over $[\ell_i, k_{i+1})$.\footnote{{When $\ell_i = k_{i+1}$, FPI is never used over $[k_i, k_{i+1})$.}} Since $\alpha^\star_\tau$ is the fixed point of $F_\tau$, we have
\begin{align*}
\| \alpha^{k_i} - \alpha^\star_\tau \|_\infty &\leq \| g^{k_i} \|_\infty +\| F_\tau\alpha^{k_i} - \alpha^\star\|_\infty \\
&\leq \| g^{k_i} \|_\infty + \gamma \| \alpha^{k_i} - \alpha^\star \|_\infty.
\end{align*}
For every $i$, the residual $g^{k_i}$ passes through the conventional safeguarding (line 16 of Algorithm~\ref{alg:aa-qmdp}). Thus, we have
\begin{align*}
 \| \alpha^{k_i} - \alpha^\star \|_\infty &\leq \frac{1}{1-\gamma} \| g^{k_i} \|_\infty \\
&\leq \frac{D\| g^0\|_\infty}{1 - \gamma} (i + 1)^{-(1 + \phi)} \xrightarrow{i \to \infty} 0.
\end{align*}
That is, $\alpha^{k_i}$ converges to $\alpha^\star_\tau$ as $i \to \infty$.

We now claim that for general $k$, the size of $\| \alpha^k - \alpha^\star \|_\infty$ is controlled appropriately so that the convergence of $\alpha^k$ to $\alpha^\star_\tau$ is achieved. To show this, suppose that an integer $k \geq k_0$ is arbitrarily chosen. Then, there exists $i \geq 0$ such that $k_i \leq k < k_{i+1}$. If $k$ lies in $[k_i, \ell_i)$, then for all $k_i \leq m \leq k$, AA is used to generate $\alpha^{m+1}$. Hence, denoting $N = |S||A|$, the following inequality holds:
\begin{align*}
\|g^{m+1}\|_\infty &\leq \|\alpha^{m+1} - F\alpha^{m+1} -\alpha^m + F\alpha^m \|_\infty +\|g^m \|_\infty \\
&\leq (1+\gamma) \|\alpha^{m+1} - \alpha^m \|_\infty + \|g^m \|_\infty \\
&= (1+\gamma)\|A_m g^m \|_\infty + \|g^m \|_\infty\\
&\leq (1+\gamma)\sqrt{N} \underbrace{\|A_m \|_2}_{\leq 1 + 2/\eta} \| g^m\|_\infty + \|g^m \|_\infty \\
&\leq \underbrace{(1+(1+\gamma)\sqrt{N}(1+2/\eta))}_{:= B}\| g^m\|_\infty \\
&= B \| g^m\|_\infty,
\end{align*}
where we use the bound~\eqref{eq:aa-mat-bound} to obtain the last inequality. Applying the above inequality for all $k_i \leq m \leq k$, we obtain
\begin{equation}\label{eq:aa-residual-bound}
\|g^k \|_\infty \leq B^{k - k_i}\|g^{k_i} \|_\infty, \quad k_i \leq k < \ell_i.
\end{equation}
For all $k_i \leq m \leq k$, we also have
\begin{align*}
\|a^{m+1} - \alpha^\star_\tau \|_\infty &= \| \alpha^m - A_m g^m - \alpha^\star_\tau \|_\infty \\
&\leq \| \alpha^m - \alpha^\star_\tau \|_\infty + \| A_m g^m \|_\infty \\
&\leq \| \alpha^m -\alpha^\star_\tau \|_\infty + \| A_m g^m \|_2 \\
&\leq \| \alpha^m -\alpha^\star_\tau \|_\infty + \sqrt{N}\| A_m \|_2 \| g^m \|_\infty \\
&\leq \| \alpha^m -\alpha^\star_\tau \|_\infty + \underbrace{\sqrt{N}(1+ 2/\eta)}_{:= C} \| g^m \|_\infty \\
&\leq \| \alpha^m -\alpha^\star_\tau \|_\infty + C \| g^m \|_\infty.
\end{align*}
Combining this inequality with~\eqref{eq:aa-residual-bound} yields that for all $k_i \leq k < \ell_i$,
\begin{align*}
\|\alpha^{k+1} - \alpha^\star_\tau \|_\infty &\leq  \| \alpha^{k_i} - \alpha^\star_\tau \|_\infty + C\sum_{j=k^{i}}^{k} \| g^j \|_\infty \\
&\leq \| \alpha^{k_i} - \alpha^\star_\tau \|_\infty + E^k_i \| g^{k_i}\|_\infty,
\end{align*}
where $E_i^k$ is a constant, defined as
\begin{equation*}
E^k_i := C\sum_{j=0}^{k - k_i} B^{j}.
\end{equation*}
In particular, by letting $k = \ell_{i}-1$, we arrive at
\begin{equation*}
\|\alpha^{\ell_i} - \alpha^\star_\tau \|_\infty \leq  \| \alpha^{k_i} - \alpha^\star_\tau \|_\infty + E^{\ell_i-1}_i \| g^{k_i}\|_\infty.
\end{equation*}
Recall that FPI is used over $[\ell_i, k_{i+1})$.
Thus, for any integer $k \in [\ell_i, k_{i+1})$, the following inequalities hold:
\begin{align*}
\| \alpha^k - \alpha^\star_\tau\|_\infty &\leq \gamma^{k - \ell_i} \| \alpha^{\ell_i} - \alpha^\star_\tau\|_\infty \\
&\leq \gamma^{k - \ell_i} \left( \| \alpha^{k_i} - \alpha^\star_\tau \|_\infty + E^{\ell_i-1}_i \|g^{k_i} \|_\infty \right) \\
&\leq \| \alpha^{k_i} - \alpha^\star_\tau \|_\infty + E^{\ell_i-1}_i \|g^{k_i} \|_\infty.
\end{align*}
However, the set $\{ E_i^k \}_{i \geq 0, k_i \leq k < \ell_i}$ is bounded from above by some constant $E$. Indeed, because $\{ k_i\}$ and $\{\ell_i \}$ are constructed to satisfy $ \ell_i - k_i \leq k_{i+1} - k_i \leq N_s$ for all $i$, we have
\begin{equation*}
E_i^k  \leq E := C\sum_{j=0}^{N_s} B^{j}, \quad i \geq 0, \quad k_i \leq k < \ell_i.
\end{equation*}
Therefore, we have
\begin{equation*}
\| \alpha^{k} - \alpha^\star_\tau\|_\infty \leq \| \alpha^{k_i} - \alpha^\star_\tau \|_\infty + E \| g^{k_i} \|_\infty, \quad k_i \leq k < k_{i+1}.
\end{equation*}
Letting $k\to\infty$, the right-hand side  goes to $0$ as both $\| \alpha^{k_i} - \alpha^\star_\tau \|_\infty$ and $ \| g^{k_i} \|_\infty$ converge to $0$.
Therefore, we conclude that
$\lim_{k\to\infty} \|\alpha^k - \alpha^\star_\tau \|_\infty = 0$ as desired.
{}
\end{proof}

After obtaining $\alpha_\tau^\star$ from Algorithm~\ref{alg:aa-qmdp}, the soft QMDP policy can be constructed as in~\eqref{eq:state-space-policy}.

\section{Simulation Method}\label{sec:sim}

So far, we have assumed that the POMDP model is available. 
Specifically, 
Algorithm~\ref{alg:aa-qmdp} requires  information about $T,O$ and $R$.
When the model information is unavailable, 
simulation-based methods can be employed to compute an approximate solution 
using a POMDP simulator (e.g.,~\cite{Silver2010,ye2017despot}).

We adopt the setting in~\cite{Silver2010}, where 
 a black-box simulator $\mathcal{G}$ is used as a \emph{generative} model of the POMDP.
 Given a state and action pair $(s, a)$, 
 the simulator generates a sample of successive state, observation and reward $\{(s'^{(j)},z^{(j)}, r^{(j)})\}_{j=1}^J$, that is, 
 $(s'^{(j)},z^{(j)}, r^{(j)}) \sim \mathcal{G} (s, a)$. 
 Using the sample, we consider the following 
empirical version of the soft QMDP operator~\eqref{eq:qmdp-maxent}:
 \begin{equation}\label{eq:qmdp-sim}
\begin{split}
&(\hat{F}_\tau \alpha) (s, a)  = \frac{1}{J} \bigg[ \sum^{J}_{j=1} r^{(j)} + \gamma \sum_{j=1}^J  \tau \ln \bigg ( \sum_{a' \in A} \exp \frac{\alpha(s'^{(j)}, a')}{\tau} \bigg ) \bigg].
\end{split}
\end{equation}
 The residual function for the operator $\hat{F}$ is  defined as
 \[
 \hat{G}_\tau (\alpha) := \alpha - \hat{F}_\tau \alpha.
 \]
 Algorithm~\ref{alg:aa-qmdp} is then accordingly modified by replacing $F_\tau$ and $G_\tau$ with $\hat{F}_\tau$ and $\hat{G}_\tau$, respectively. 
 However, the simulation-based method is not guaranteed to converge due to simulation errors. 
Instead, we analyze an error bound for the simulation-based method by examining how the simulation errors are propagated through the AA-QMDP algorithm.

For further analysis, we consider the following mismatch between $\hat{F} \alpha^k$ and $F \alpha^k$:
\[
e^k := \hat{F}_\tau \alpha^k - F_\tau \alpha^k,
\]
which can be interpreted as the error caused by simulation. 
Then, the FPI candidate in Algorithm~\ref{alg:aa-qmdp} can be expressed as
\[
\alpha^{k+1} = \hat{F}_\tau \alpha^k = F_\tau \alpha^k + e^k.
\]
We show that the residual $G(\alpha^k) = \alpha^k -  F\alpha^k$ is bounded whenever the size of the error $e^k$ is bounded.
\begin{theorem}\label{thm:sim}
Consider Algorithm~\ref{alg:aa-qmdp} with the simulation-based soft QMDP operator $\hat{F}_\tau$.
Suppose there exists a positive constant $\delta$ such that  $\|e^k\|_\infty \leq \varepsilon$ for all $k$.
Then, the following error bound holds:
\[
\liminf_{k\to\infty} \|G_\tau(\alpha^k)\|_\infty \leq \frac{1+\gamma}{1-\gamma}\varepsilon.	
\]
\end{theorem}
\begin{proof}
We first assume that Algorithm~\ref{alg:aa-qmdp} selects AA candidates for only finitely many times. Then, there exists some $K$ such that for all $k \geq K$, $\alpha^k$ are generated via FPI. In this case, Algorithm~\ref{alg:aa-qmdp} reduces to FPI with the initial iterate $\alpha^{K}$. Thus, it suffices to prove that the simulation-based FPI satisfies the error bound. To show this, we first notice that
\begin{equation*}
\Vert F_\tau\alpha^k - \alpha^\star \Vert_\infty = \Vert F_\tau \alpha^k - F_\tau \alpha^\star \Vert_\infty \leq \gamma\Vert \alpha^k - \alpha^\star \Vert_\infty,
\end{equation*}
where the last inequality holds since $F_\tau$ is a $\gamma$-contraction.
Thus,  we have
\begin{equation}\label{eq:proof-sim-ineq1}
\begin{split}
\Vert G_\tau(\alpha^k) \Vert_\infty &= \Vert \alpha^k - \alpha^\star + \alpha^\star - F_\tau\alpha^k\Vert_\infty\\
& \leq (1+\gamma)\Vert \alpha^k - \alpha^\star \Vert_\infty.
\end{split}
\end{equation}
To derive an upper bound of $\Vert \alpha^k - \alpha^\star \Vert_\infty$, we observe that
\begin{align*}
\Vert \alpha^{k+1} - \alpha^\star \Vert_\infty &\leq \Vert \underbrace{\hat F_\tau\alpha^k - F_\tau\alpha^k}_{=e^k} \Vert_\infty + \Vert F_\tau\alpha^k - \alpha^\star \Vert_\infty \\
&\leq \varepsilon + \gamma\Vert \alpha^k - \alpha^\star\Vert_\infty.
\end{align*}
Recursively applying this inequality yields
\begin{equation*}
\Vert \alpha^k - \alpha^\star \Vert_\infty \leq \varepsilon(1 + \gamma + \cdots \gamma^{k-1}) + \gamma^k \Vert \alpha^0 - \alpha^\star \Vert_\infty.
\end{equation*}  
Therefore, we have
\begin{equation}\label{eq:proof-sim-ineq2}
\liminf_{k\to\infty} \Vert \alpha^k - \alpha^\star \Vert_\infty \leq \frac{\varepsilon}{1 - \gamma}.
\end{equation}
Combining \eqref{eq:proof-sim-ineq1} and \eqref{eq:proof-sim-ineq2}, we obtain
\begin{equation*}
\liminf_{k\to\infty} \Vert G_\tau(\alpha^k) \Vert_\infty \leq \frac{1+\gamma}{1 - \gamma} \varepsilon.
\end{equation*}

Suppose now that AA candidates are selected infinitely many times. 
Recall that $g^k = \hat{G} (\alpha^k)$. 
Thus, 
\begin{align*}
\Vert G_\tau(\alpha^{k})  - g^{k} \Vert_\infty &= \Vert G_\tau(\alpha^{k})  - \hat G_\tau(\alpha^{k}) \Vert_\infty \\
&= \Vert \underbrace{F(\alpha^{k})  - \hat F(\alpha^{k})}_{=e^k} \Vert_\infty \leq \varepsilon.
\end{align*}
Let $k_i$ denote each initial iteration count for accepting an AA candidate.
Then, the set of $k_i$’s has infinitely many elements, and we have
\begin{align} \nonumber
\liminf_{k\to\infty} \Vert G_\tau(\alpha^k) \Vert_\infty 
&\leq \liminf_{i\to\infty} \Vert G_\tau(\alpha^{k_i}) \Vert_\infty \\\nonumber
&\leq \liminf_{i\to\infty} (\Vert G_\tau(\alpha^{k_i})  - g^{k_i} \Vert_\infty + \Vert g^{k_i} \Vert_\infty) \\\nonumber
&\leq \varepsilon +  \liminf_{i\to\infty} \Vert g^{k_i} \Vert_\infty \\\nonumber
&\leq \varepsilon + D\Vert g^0 \Vert_\infty \lim_{i \to\infty} (i+1)^{-(1+\phi)} \\\nonumber
&= \varepsilon \leq \frac{1+\gamma}{1-\gamma} \varepsilon,
\end{align}
where we use $\| g^{k_i} \|_\infty \leq D \| g^0\|_\infty (i+1)^{-(1+\phi)}$
in the fourth inequality.
{}
\end{proof}

The error bound in
Theorem~\ref{thm:sim} is linear in $\varepsilon$ and depends only on $\gamma$ and $\varepsilon$. 
This implies that AA and entropy regularization  affect the error bound only through $\varepsilon$, which is the error caused by simulations.
The deviation of $\alpha^k$ from the fixed point $\alpha^\star$ of the original QMDP operator $F$ can also be bounded as follows. 

\begin{cor}\label{cor:sim}
Consider Algorithm~\ref{alg:aa-qmdp} with the simulation-based soft QMDP operator $\hat{F}_\tau$.
Suppose there exists a positive constant $\delta$ such that  $\|e^k\|_\infty \leq \varepsilon$ for all $k$.
Then, the following error bound holds:
\begin{equation}\label{eq:alpha-sim-error}
\liminf_{k\to\infty}\|\alpha^k - \alpha^\star\|_\infty \leq \frac{ \gamma\tau\ln|A| + (1+\gamma)\varepsilon}{(1-\gamma)^2}.
\end{equation}
\end{cor}
\begin{proof}
It follows from Proposition~\ref{pro:alpha-bound} that
\begin{equation}\label{eq1}
\|\alpha^\star - \alpha^\star_\tau\|_\infty \leq \frac{\gamma\tau \ln|A|}{1-\gamma}.
\end{equation}
At the $k$th iteration of the algorithm ($k\geq 0$), we have
\begin{equation}\label{eq2}
\begin{split}
\|\alpha^\star - \alpha^k\|_\infty 
&= \|\alpha^\star -\alpha^\star_\tau + F_\tau\alpha^\star_\tau - F_\tau \alpha^k
+ F_\tau \alpha^k - \alpha^k\|_\infty \\
&\leq \|\alpha^\star -\alpha^\star_\tau + F_\tau\alpha^\star_\tau - F_\tau \alpha^k \|_\infty
+ \|F_\tau \alpha^k - \alpha^k\|_\infty \\
&\leq \|\alpha^\star -\alpha^\star_\tau\|_\infty + \gamma \|\alpha^\star_\tau - \alpha^k\|_\infty
+ \|F_\tau \alpha^k - \alpha^k\|_\infty,
\end{split}
\end{equation}
where the last inequality follows from the $\gamma$-contraction property of $F_\tau$. 
Combining~\eqref{eq1} and~\eqref{eq2} yields
\begin{equation*}
\begin{split}
(1-\gamma) \|\alpha^\star - \alpha^k\|_\infty 
&\leq \|\alpha^\star - \alpha^\star_\tau\|_\infty + \|\alpha^k - F_\tau \alpha^k\|_\infty \\
&\leq \|\alpha^\star - \alpha^\star_\tau\|_\infty + \|G_\tau(\alpha^k)\|_\infty\\
&\leq \frac{\gamma\tau\ln|A|}{1-\gamma} + \|G_\tau(\alpha^k)\|_\infty.
\end{split}
\end{equation*}
Taking $\liminf$ on both sides and using Theorem~\ref{thm:sim}, we obtain
\begin{equation*}
\begin{split}
(1-\gamma) \liminf_{k\to\infty}\|\alpha^k - \alpha^\star\|_\infty 
&\leq \frac{\gamma\tau\ln|A|}{1-\gamma} + \liminf_{k\to\infty}\|G_\tau(\alpha^k)\|_\infty \\
&\leq \frac{\gamma\tau\ln|A|}{1-\gamma} + \frac{1+\gamma}{1-\gamma}\varepsilon.
\end{split}
\end{equation*}
Dividing both sides by $1-\gamma$, the result follows. 
{}
\end{proof}

\section{Numerical Experiments}\label{sec:num}
In this section, we empirically evaluate the performance of the proposed algorithm on several benchmark problems and analyze the results.

\subsection{The Setup}
Our experiments use popular robot navigation benchmark problems\footnote{\label{footnote}https://www.pomdp.org/examples}, namely \emph{sunysb, fourth, tag}, and \emph{underwater}, each of which is modeled as a finite POMDP.
In these environments, the robot is unable to observe its location and
 has two types of actions: {\tt declare-goal} and {\tt navigate}.
The robot has to first navigate to the goal position using {\tt navigate} action and then take {\tt declare-goal} action due to its inability to localize itself.
In \emph{underwater}, the robot has an additional action called {\tt estimate} by which the robot collects information about the environment to localize itself.
Depending on the environment, one of the following two reward functions is used: The first  penalizes the robot in each time step until it reaches the goal position and chooses \emph{declare-goal},
and the second gives a positive reward only if the robot performs \emph{declare-goal} at the right position; no reward is given for all the other the state-action pairs.
Interestingly, the \emph{tag} environment has a dynamic goal position, namely, the position of the opponent running away from the robot. Each of the other environments has a static goal position.

As a baseline, we use SARSOP~\cite{kurniawati2008sarsop},
a state-of-the art point-based offline POMDP algorithm.
SARSOP is a tree search-based method, and thus the quality of its solutions heavily depends on the amount of  allocated computational resources.
For a fair comparison in terms of the total reward, 
the computation time allowed for the tree search is set to be the minimum computation time for AA-sQMDP.

The benchmark problems are solved with 100 different initializations.
The total reward is then averaged over the 100 trajectories of maximum  length 100.
Each method is tested with two different choices of the initial beliefs: 
a pre-defined fixed initial belief and a randomly chosen belief from $\mathcal{B}$.
The following metrics are used to quantify the performance of each algorithm:
\begin{itemize}
\item $\#\mathrm{iter}$: the minimum number of iterations required to meet the stopping condition $\|g^k\|_\infty < 10^{-6}$,

\item $\#\mathrm{AA}$: the number of iterations in which the AA candidate is selected,

\item $t_{\mathrm{total}}$(sec): the total computation time,
 
\item $\mathrm{reward}_{\mathrm{fixed}}$: the total reward evaluated with a fixed initial belief,

\item $\mathrm{reward}_{\mathrm{rand}}$: the total reward evaluated with a randomly chosen initial belief.

\end{itemize}
All the experiments were conducted on a PC with Intel Core i7-8700K at 3.70GHz. 
The source code of our AA-sQMDP implementation is available online.\footnote{https://github.com/CORE-SNU/AA-POMDP}

\begin{table*}[tb]
\caption{Experiment results: Model-based case ($\mathrm{mean} \pm \mathrm{std}$).}
\begin{center}
\scalebox{0.53}{
\begin{tabular}{|l|l|c|c|c|c|c|c|c|}
\hline
& & \multicolumn{7}{c|}{\textbf{Algorithm}}  \\  \hline
\multirow{2}{*}{\vtop{\hbox{\strut \textbf{Problem}}\hbox{\strut \textbf{($|\mathcal{S}|, |\mathcal{A}|, |\mathcal{O}|$)}}}} 
& \multirow{2}{*}{\textbf{Metric}} & \multirow{2}{*}{\textbf{QMDP}} & \multicolumn{2}{c|}{\textbf{AA-QMDP}}
& \multirow{2}{*}{\textbf{sQMDP}} & \multicolumn{2}{c|}{\textbf{AA-sQMDP}} & \multirow{2}{*}{\textbf{SARSOP}} \\
& & & without $\theta_\mathrm{targ}$ & with $\theta_\mathrm{targ}$ & & without $\theta_\mathrm{targ}$ & with $\theta_\mathrm{targ}$ & \\\hline \hline

\multirow{5}{*}{\vtop{\hbox{\strut \emph{sunysb}}\hbox{\strut ($300,4,28$)}}}
&  $\#\mathrm{iter}$ & 1364.06 $\pm$ 13.59 & 509.63 $\pm$ 143.45 & \textbf{120.44 $\pm$ 10.14} & 2095.00 $\pm$ 0.00 & 102.50 $\pm$ 2.19 & \textbf{91.54 $\pm$ 2.18} & N/A \\
&  $\#\mathrm{AA}$ & N/A & 508.63 $\pm$ 143.45 & 23.61 $\pm$ 3.57 & N/A & 101.50 $\pm$ 2.19 & 70.53 $\pm$ 1.98 & N/A \\
&  $t_{\mathrm{total}}$(sec) & 0.169 $\pm$ 0.002 & 0.258 $\pm$ 0.073 & \textbf{0.061 $\pm$ 0.006} & 1.022 $\pm$ 0.002 & 0.090 $\pm$ 0.003 & \textbf{0.081 $\pm$ 0.003} & 0.081 \\
&  $\mathrm{reward}_{\mathrm{fixed}}$ & 0.759 $\pm$ 0.016 & 0.760 $\pm$ 0.016 & 0.761 $\pm$ 0.017 & 0.767 $\pm$ 0.014 & 0.765 $\pm$ 0.017 & 0.763 $\pm$ 0.015 & 0.434 $\pm$ 0.233 \\
&  $\mathrm{reward}_{\mathrm{rand}}$ & 0.377 $\pm$ 0.040 & 0.388 $\pm$ 0.043 & 0.383 $\pm$ 0.042 & 0.317 $\pm$ 0.040 & 0.305 $\pm$ 0.039 & 0.316 $\pm$ 0.040 & 0.020 $\pm$ 0.016  \\ 
\hline

\multirow{5}{*}{\vtop{\hbox{\strut \emph{fourth}}\hbox{\strut ($1052,4,28$)}}}
&  $\#\mathrm{iter}$ & 1362.62 $\pm$ 13.05 & 1119.73 $\pm$ 266.48 & \textbf{257.11 $\pm$ 32.16} & 2095.00 $\pm$ 0.00 & 105.26 $\pm$ 1.23 & \textbf{98.42 $\pm$ 2.00} & N/A \\
&  $\#\mathrm{AA}$ & N/A & 1118.73 $\pm$ 266.48 & 30.27 $\pm$ 9.72 & N/A & 104.26 $\pm$ 1.23 & 71.00 $\pm$ 1.66 & N/A \\
&  $t_{\mathrm{total}}$(sec) & 0.241 $\pm$ 0.003 & 0.907 $\pm$ 0.216 & \textbf{0.206 $\pm$ 0.026} & 2.701 $\pm$ 0.011 & 0.208 $\pm$ 0.004 & \textbf{0.196 $\pm$ 0.005} & 0.196 \\
&  $\mathrm{reward}_{\mathrm{fixed}}$ & 0.592 $\pm$ 0.013 & 0.592 $\pm$ 0.010 & 0.592 $\pm$ 0.011 & 0.601 $\pm$ 0.009 & 0.505 $\pm$ 0.212 & 0.523 $\pm$ 0.193 & \textbf{fail} \\
&  $\mathrm{reward}_{\mathrm{rand}}$ & 0.371 $\pm$ 0.030 & 0.380 $\pm$ 0.034 & 0.372 $\pm$ 0.038 & 0.314 $\pm$ 0.035 & 0.278 $\pm$ 0.031 & 0.276 $\pm$ 0.035 & \textbf{fail} \\ 
\hline

\multirow{5}{*}{\vtop{\hbox{\strut \emph{tag}}\hbox{\strut ($870,5,30$)}}}
&  $\#\mathrm{iter}$ & 315.62 $\pm$ 0.49 & \textbf{87.58 $\pm$ 8.49} & 105.27 $\pm$ 8.70 & 414.00 $\pm$ 0.00 & 62.62 $\pm$ 1.35 & \textbf{58.16 $\pm$ 1.41} & N/A \\
&  $\#\mathrm{AA}$ & N/A & 86.58 $\pm$ 8.49 & 56.65 $\pm$ 7.37 & N/A & 61.62 $\pm$ 1.35 & 40.58 $\pm$ 1.21 & N/A \\
&  $t_{\mathrm{total}}$(sec) & \textbf{0.072 $\pm$ 0.001} & 0.085 $\pm$ 0.009 & 0.101 $\pm$ 0.009 & 0.691 $\pm$ 0.003 & 0.157 $\pm$ 0.004 & \textbf{0.144 $\pm$ 0.005} & 0.144 \\
&  $\mathrm{reward}_{\mathrm{fixed}}$ & $-15.932$ $\pm$ 0.696 & $-16.039$ $\pm$ 0.773 & $-16.001$ $\pm$ 0.670 & $-6.643$ $\pm$ 0.638 & $-6.699$ $\pm$ 0.630 & $-6.735$ $\pm$ 0.628 & $-18.961$ $\pm$ 0.466 \\
&  $\mathrm{reward}_{\mathrm{rand}}$ & $-15.695$ $\pm$ 0.908 & $-15.627$ $\pm$ 0.985 & $-15.445$ $\pm$ 1.040 & $-6.389$ $\pm$ 0.654 & $-6.441$ $\pm$ 0.668 & $-6.351$ $\pm$ 0.616 & $-18.316$ $\pm$ 0.601 \\ 
\hline

\multirow{5}{*}{\vtop{\hbox{\strut \emph{underwater}}\hbox{\strut ($2653,6,102$)}}}
&  $\#\mathrm{iter}$ & 446.21 $\pm$ 5.68 & 240.14 $\pm$ 14.96 & \textbf{78.50 $\pm$ 13.78} & 441.54 $\pm$ 5.11 & 165.98 $\pm$ 6.10 & \textbf{132.53 $\pm$ 11.63} & N/A \\
&  $\#\mathrm{AA}$ & N/A & 239.14 $\pm$ 14.96 & 2.29 $\pm$ 1.54 & N/A & 164.98 $\pm$ 6.10 & 22.82 $\pm$ 2.85 & N/A \\
&  $t_{\mathrm{total}}$(sec) & 0.202 $\pm$ 0.003 & 0.650 $\pm$ 0.044 & \textbf{0.201 $\pm$ 0.038} & 2.197 $\pm$ 0.024 & 1.260 $\pm$ 0.047 & \textbf{1.005 $\pm$ 0.093} & 1.005 \\
&  $\mathrm{reward}_{\mathrm{fixed}}$ & $-44.353$ $\pm$ 42.644 & $-43.957$ $\pm$ 47.155 & $-45.929$ $\pm$ 44.507 & 681.420 $\pm$ 8.228 & 680.654 $\pm$ 7.600 & 681.448 $\pm$ 8.456 & 688.200 $\pm$ 8.081 \\
&  $\mathrm{reward}_{\mathrm{rand}}$ & 3307.193 $\pm$ 248.000 & 3301.785 $\pm$ 220.116 & 3289.063 $\pm$ 237.469 & 3230.717 $\pm$ 267.828 & 3243.995 $\pm$ 261.686 & 3251.146 $\pm$ 241.516 & 3273.417 $\pm$ 228.879 \\
\hline

\end{tabular}
}
\label{tab:aa-pomdp-model}
\end{center}
\end{table*}

\subsection{Model-Based Method}

Table~\ref{tab:aa-pomdp-model} presents the numerical results of AA-sQMDP and SARSOP in the standard model-based setting. We also compare our algorithm to QMDP, soft QMDP (or sQMDP), and AA-$(s)$QMDP that does not use the target acceleration factor. In particular, the results obtained using the safeguarding rule based on the target acceleration factor are shown in the columns ``with $\theta_\mathrm{targ}$", and the columns marked  ``without $\theta_\mathrm{targ}$" present the results of the safeguarded AA only using the target residual gain.
The hyperparameters of the algorithms were optimized within the following ranges:
$\bar{m} = 1$, $m \in \{10^{-2}, 10^0, 10^2, 10^4\}$, $\tau  \in \{10^1, 10^3, 10^5\}$.
Additional hyperparameter values were chosen as follows:
$D = 10^6$, $\varepsilon = 10^{-6}$,
$N_s = 400$, $\eta = 10^{-16}$, and 
$\mmax = 16$.

As shown in Table~\ref{tab:aa-pomdp-model}, using AA significantly reduces the number of iterations to solve all the POMDP problems. However, as each AA sequence converges to the same fixed point as that of FPI, both $\mathrm{reward}_{\mathrm{fixed}}$ and $\mathrm{reward}_{\mathrm{rand}}$ remain unchanged or are affected to a small extent when AA is used. For example, in \emph{tag} environment, the last three methods that uses soft QMDP attain $\mathrm{reward}_{\mathrm{fixed}}$ around $-6.7$ although the three methods differ from each others in the use of AA and the target acceleration factor. This observation is consistent with Theorem~\ref{thm:global-conv}, where the global convergence of AA-sQMDP to the correct fixed point is proved.

Furthermore, the effect of the maximum entropy regularization is clearly observed when comparing the results of AA-QMDP and AA-sQMDP.
Regardless of the use of the target acceleration factor, AA-sQMDP requires significantly fewer  iterations until convergence than AA-QMDP while also maintaining higher rewards.\footnote{While the required number of iterations is saved to a large extent by using the maximum entropy regularization, the actual computation time may not be reduced. This is because when using a large $\tau$, arithmetic operations involved in a single update of soft QMDP have to be executed on numbers of large scale, as pointed out in Appendix~\ref{appx:mellowmax}. This can be relieved by employing the KL divergence regularization introduced in Appendix~\ref{appx:mellowmax} instead of the maximum entropy regularization. For instance, see Table~\ref{tab:aa-pomdp-mellow}.} However, such a tendency is not observed in \emph{sunysb} and \emph{fourth} when AA is not used. This implies that the benefit of using the maximum entropy method comes from the quasi-Newton nature of AA. It is also notable that in \emph{tag} and \emph{underwater}, the maximum entropy regularization greatly benefits the total reward.
In fact, using a large $\tau$  benefits POMDPs in which the uncertainty of the belief lasts many time steps: In \emph{tag}, the agent does not know the position of the moving target, and in \emph{underwater}, the agent is unable to identify the exact position of itself.

We also confirm that using the target acceleration factor increases the efficiency of AA, as few AA steps are rejected via safeguarding without $\theta_{\mathrm{targ}}$. However, our novel safeguarding rule with $\theta_{\mathrm{targ}}$ effectively excludes many unsuccessful AA steps. As a result, the employment of the target acceleration factor  largely reduces the number of iterations needed for the convergence of AA-sQMDP.
Notably, in \emph{sunysb}, using AA saves $93\%$ of $\#\mathrm{iter}$ required for FPI.

Finally, in terms of the cumulative reward, AA-sQMDP is better or as good as or better than SARSOP in all environments.
However, the computation of an AA-sQMDP policy~\eqref{eq:state-space-policy} is incomparably faster than the tree search procedure of SARSOP. This is one of the key advantages of our offline method compared to the popular online methods, making the use of our method feasible in practice. To summarize, the policies computed offline by AA-sQMDP  are able to solve nontrivial POMDP problems in practice, achieving a state-of-the-art performance even without costly online computation or search processes.

\begin{table*}[tb]
\caption{Experiment results: Simulation-based case ($\mathrm{mean} \pm \mathrm{std}$). $D = 10$ is used in all problems.}
\begin{center}
\scalebox{0.6}{
\begin{tabular}{|l|l|c|c|c|c|c|c|}
\hline
& & \multicolumn{6}{c|}{\textbf{Algorithm}}  \\  \hline
\multirow{2}{*}{\vtop{\hbox{\strut \textbf{Problem}}\hbox{\strut \textbf{($|\mathcal{S}|, |\mathcal{A}|, |\mathcal{O}|$)}}}} 
& \multirow{2}{*}{\textbf{Metric}} & \multirow{2}{*}{\textbf{QMDP}} & \multicolumn{2}{c|}{\textbf{AA-QMDP}}
& \multirow{2}{*}{\textbf{sQMDP}} & \multicolumn{2}{c|}{\textbf{AA-sQMDP}} \\
& & & without $\theta_\mathrm{targ}$ & with $\theta_\mathrm{targ}$ & & without $\theta_\mathrm{targ}$ & with $\theta_\mathrm{targ}$ \\\hline \hline

\multirow{5}{*}{\vtop{\hbox{\strut \emph{sunysb}}\hbox{\strut ($300,4,28$)}}}
& $\#\mathrm{iter}$ & 1361.51 $\pm$ 16.30 & 627.58 $\pm$ 184.63 & \textbf{140.08 $\pm$ 19.63} & 2095.00 $\pm$ 0.00 & 104.59 $\pm$ 2.08 & \textbf{93.77 $\pm$ 2.62} \\
& $\#\mathrm{AA}$ & N/A & 626.58 $\pm$ 184.63 & 46.69 $\pm$ 11.81 & N/A & 103.59 $\pm$ 2.08 & 72.61 $\pm$ 2.24 \\
& $t_{\mathrm{total}}$(sec) & 0.629 $\pm$ 0.011 & 0.567 $\pm$ 0.166 & \textbf{0.132 $\pm$ 0.020} & 1.756 $\pm$ 0.013 & 0.143 $\pm$ 0.010 & \textbf{0.130 $\pm$ 0.010} \\
& $\mathrm{reward}_{\mathrm{fixed}}$ & 0.744 $\pm$ 0.030 & 0.744 $\pm$ 0.033 & 0.739 $\pm$ 0.035 & 0.769 $\pm$ 0.013 & 0.766 $\pm$ 0.013 & 0.763 $\pm$ 0.017 \\
& $\mathrm{reward}_{\mathrm{rand}}$ & 0.377 $\pm$ 0.050 & 0.380 $\pm$ 0.046 & 0.390 $\pm$ 0.052 & 0.298 $\pm$ 0.040 & 0.301 $\pm$ 0.044 & 0.300 $\pm$ 0.048 \\ 
\hline

\multirow{5}{*}{\vtop{\hbox{\strut \emph{fourth}}\hbox{\strut ($1052,4,28$)}}}
& $\#\mathrm{iter}$ & 1361.53 $\pm$ 15.45 & 1325.61 $\pm$ 290.75 & \textbf{274.90 $\pm$ 28.67} & 2095.00 $\pm$ 0.00 & 107.67 $\pm$ 1.82 & \textbf{99.54 $\pm$ 2.21} \\
& $\#\mathrm{AA}$ & N/A & 1324.61 $\pm$ 290.75 & 38.87 $\pm$ 12.24 & N/A & 106.67 $\pm$ 1.82 & 71.76 $\pm$ 1.90 \\
& $t_{\mathrm{total}}$(sec) & 0.732 $\pm$ 0.011 & 1.596 $\pm$ 0.353 & \textbf{0.340 $\pm$ 0.037} & 3.481 $\pm$ 0.017 & 0.263 $\pm$ 0.009 & \textbf{0.246 $\pm$ 0.010} \\
& $\mathrm{reward}_{\mathrm{fixed}}$ & 0.597 $\pm$ 0.010 & 0.595 $\pm$ 0.011 & 0.593 $\pm$ 0.012 & 0.596 $\pm$ 0.010 & 0.512 $\pm$ 0.205 & 0.511 $\pm$ 0.232 \\
& $\mathrm{reward}_{\mathrm{rand}}$ & 0.389 $\pm$ 0.037 & 0.390 $\pm$ 0.035 & 0.386 $\pm$ 0.030 & 0.311 $\pm$ 0.035 & 0.270 $\pm$ 0.032 & 0.273 $\pm$ 0.041 \\ 
\hline

\multirow{5}{*}{\vtop{\hbox{\strut \emph{tag}}\hbox{\strut ($870,5,30$)}}}
& $\#\mathrm{iter}$ & 315.63 $\pm$ 0.48 & \textbf{88.44 $\pm$ 10.29} & 103.68 $\pm$ 9.56 & 414.00 $\pm$ 0.00 & 62.69 $\pm$ 1.29 & \textbf{58.17 $\pm$ 1.18} \\
& $\#\mathrm{AA}$ & N/A & 87.44 $\pm$ 10.29 & 57.04 $\pm$ 9.09 & N/A & 61.69 $\pm$ 1.29 & 40.56 $\pm$ 1.02 \\
& $t_{\mathrm{total}}$(sec) & 0.223 $\pm$ 0.002 & \textbf{0.131 $\pm$ 0.020} & 0.150 $\pm$ 0.018 & 0.900 $\pm$ 0.008 & 0.191 $\pm$ 0.011 & \textbf{0.176 $\pm$ 0.008} \\
& $\mathrm{reward}_{\mathrm{fixed}}$ & $-16.507$ $\pm$ 0.985 & $-16.563$ $\pm$ 0.945 & $-16.637$ $\pm$ 1.042 & $-6.734$ $\pm$ 0.679 & $-6.859$ $\pm$ 0.606 & $-6.777$ $\pm$ 0.607 \\
& $\mathrm{reward}_{\mathrm{rand}}$ & $-15.514$ $\pm$ 1.220 & $-15.481$ $\pm$ 1.188 & $-15.648$ $\pm$ 1.239 & $-6.483$ $\pm$ 0.591 & $-6.223$ $\pm$ 0.627 & $-6.428$ $\pm$ 0.686 \\ 
\hline

\multirow{5}{*}{\vtop{\hbox{\strut \emph{underwater}}\hbox{\strut ($2653,6,102$)}}}
& $\#\mathrm{iter}$ & 446.75 $\pm$ 3.64 & 241.72 $\pm$ 15.58 & \textbf{77.46 $\pm$ 12.75} & 441.29 $\pm$ 5.78 & 166.13 $\pm$ 6.88 & \textbf{132.52 $\pm$ 13.66} \\
& $\#\mathrm{AA}$ & N/A & 240.72 $\pm$ 15.58 & 2.30 $\pm$ 1.02 & N/A & 165.13 $\pm$ 6.88 & 22.86 $\pm$ 2.79 \\
& $t_{\mathrm{total}}$(sec) & 0.446 $\pm$ 0.006 & 0.732 $\pm$ 0.051 & \textbf{0.232 $\pm$ 0.040} & 2.464 $\pm$ 0.033 & 1.317 $\pm$ 0.056 & \textbf{1.048 $\pm$ 0.110} \\
& $\mathrm{reward}_{\mathrm{fixed}}$ & $-43.825$ $\pm$ 44.207 & $-47.085$ $\pm$ 50.013 & $-52.007$ $\pm$ 47.372 & 681.270 $\pm$ 7.105 & 681.291 $\pm$ 8.274 & 679.472 $\pm$ 8.468 \\
& $\mathrm{reward}_{\mathrm{rand}}$ & 3243.665 $\pm$ 237.957 & 3241.428 $\pm$ 254.624 & 3293.090 $\pm$ 247.493 & 3241.704 $\pm$ 258.867 & 3227.778 $\pm$ 249.721 & 3233.540 $\pm$ 245.822 \\  
\hline
\end{tabular}
}
\label{tab:aa-pomdp-sample}
\end{center}
\end{table*}

\subsection{Simulation-Based Method}
Table~\ref{tab:aa-pomdp-sample}  presents the results of using the simulation-based method with dataset size $J = 10$.
The results indicate that even without the exact model, 
similar levels of performance are achieved for
 every metric. Only small differences of rewards between Table~\ref{tab:aa-pomdp-model} and Table~\ref{tab:aa-pomdp-sample} are observed. Therefore, the simulation-based method, even with a small dataset, successfully replaces the exact method, with only a slightly diminished performance and requiring slightly more  computation time  to construct the approximate models~\eqref{eq:qmdp-sim} online.

\section{Conclusions}

 We presented AA-sQMDP, an accelerated offline PODMP method that carefully combines AA and  
  the maximum entropy regularization of  QMDP.
Our algorithm enjoys theoretical guarantees, including
the convergence to the soft QMDP solution and  provable error bounds for the simulation-based implementation. 
  The results of our numerical experiments demonstrate the superior performance of AA-sQMDP compared to its standard counterparts,
   in terms of the total number of iterations and  total computation time. 
  
In the future, the proposed method can be extended in several interesting directions, including
$(i)$ application to partially observable reinforcement learning, and 
$(ii)$  adaptive selection of the temperature parameter,  depending on how the algorithm proceeds.

\appendix
\section{Proofs for Section~\ref{sec:reg-pomdp}}
\subsection{Proof of Lemma~\ref{lem:contraction-qmdp}}\label{app:lem1}
\begin{proof}
Fix arbitrary $\alpha, \alpha' \in \R^{|S| \times |A|}$. For any $(s, a)\in S\times A$, we have
\begin{equation}\nonumber
\begin{split}
| F \alpha (s,a) - F \alpha' (s,a) | 
&\leq \gamma \sum_{s' \in S} T(s, a, s') \left | \max_{a' \in A} \alpha (s', a') - \max_{a' \in A} \alpha' (s', a') \right | \\
&\leq \gamma \sum_{s' \in S} T(s, a, s') \max_{a'\in A}| \alpha (s', a') - \alpha' (s', a') | \\
&\leq \gamma \sum_{s' \in S} T(s, a, s') \| \alpha  - \alpha'  \|_\infty \\
&\leq \gamma \| \alpha - \alpha' \|_\infty.
\end{split}
\end{equation}
Since $\alpha$ and  $\alpha'$ were arbitrarily chosen, $F$ is a $\gamma$-contraction mapping with respect to $\| \cdot \|_\infty$. {}
\end{proof}

\subsection{Proof of Proposition~\ref{pro:alpha-bound}} \label{app:prop}
\begin{proof}
Recall that the soft QMDP operator is given by
\[
F_\tau \alpha(s,a) = R(s, a) + \gamma \sum_{s' \in S} T (s, a, s') \ent^\star_{\tau}(\alpha(s', \cdot)),
\]
where
\[
\ent^\star_{\tau}(\alpha(s, \cdot)) = \max_{\phi \in \mathcal{P}(A)} \left[\sum_{a \in A} \phi(a) \alpha(s, a) + \tau \ent(\phi) \right].
\]
Since $0 \leq \ent(\phi) \leq \ln|A|$, we have
\begin{equation}\label{eq:ent-ineq}
\max_{a'\in A} \alpha(s',a') \leq \ent^\star_{\tau}(\alpha(s', \cdot)) \leq \max_{a'\in A} \alpha(s',a') + \tau \ln|A|
\end{equation}
for an arbitrary $\alpha \in \R^{|S| |A|}$ and $s' \in S$. It follows from  the second inequality of~\eqref{eq:ent-ineq} that for every $(s,a) \in S \times A$,
\[
F \alpha(s,a) \leq F_\tau \alpha(s,a) \leq F \alpha(s,a) + \gamma \tau \ln|A|.
\]
Since $F_\tau$ is a monotone operator, applying $F_\tau$ twice yields
\begin{equation*}
\begin{split}
F_\tau (F_\tau \alpha(s,a)) &\leq F_\tau(F \alpha(s,a) + \gamma \tau) \\
&\leq F_\tau(F \alpha(s,a)) + \gamma^2 \tau \ln|A|\\
&\leq F^2 \alpha(s,a) + (\gamma \tau + \gamma^2 \tau)\ln|A|.
\end{split}
\end{equation*}
In general, the following inequality is true for any $n$, which can be shown by induction:
\begin{equation*}
(F_\tau)^n \alpha (s, a) \leq F^n\alpha (s, a) + (\gamma + \cdots + \gamma^n)\tau\ln|A|.
\end{equation*}
Letting $n \to \infty$ on both sides yields
\[
\alpha^\star_\tau(s,a) \leq \alpha^\star(s,a) + \frac{\gamma\tau\ln|A|}{1-\gamma}.
\]
Similarly using the first inequality of~\eqref{eq:ent-ineq}, we can show that $\alpha^\star$ is the lower bound of $\alpha^\star_\tau$, i.e.,
\[
\alpha^\star_\tau(s,a) \geq \alpha^\star(s,a).
\]
Combining these two bounds leads to
\[
0 \leq \alpha^\star_\tau(s,a) - \alpha^\star(s,a) \leq \frac{\gamma\tau\ln|A|}{1-\gamma},
\]
and taking maximum on both sides completes the proof. {}
\end{proof}

\section{Maximum Entropy FIB}\label{appx:reg-fib}

\subsection{Fast Informed Bound Method}
FIB is another offline state-space method using a fixed number of $\alpha$-vectors~\cite{hauskrecht2000value}.
At the $k$th iteration of FIB, each $\alpha$-vector is updated as follows:
\begin{equation}\label{fib}
\begin{split}
&\alpha^{k+1}_a (s):= R(s, a) + \gamma \sum_{z} \max_{a'} \sum_{s'} T^O(s, a, s', z) \alpha^k_{a'} (s'),
\end{split}
\end{equation}
where the joint probability kernel $T^O : S \times A \to \mathcal{P}(S \times Z)$ is defined as
\begin{equation*}
T^O(s, a, s', z) := T (s, a, s') O (s', a, z).
\end{equation*}
It is well-known that the FIB update serves as an upper bound of the exact VI update~\cite{hauskrecht2000value}.
The FIB update~\eqref{fib} can be derived from the exact Bellman operator~\eqref{bellman} by decoupling the expectation over $b(s)$ from $R(b, a)$ and $\mathrm{Pr}(z | b, a)$.
Then, the cardinality of $\Gamma_k$ remains unchanged as $k$ increases, and the time complexity of each iteration of FIB is in $\mathcal{O}(|A|^2 |S|^2 |Z|)$.
The FIB update is independent of belief states similarly to the QMDP update.


In a similar fashion to the QMDP case, we define the FIB operator $F:\R^{|S| |A|}\to\R^{|S| |A|}$ as follows:
\begin{equation}\label{eq:fib}
\begin{split}
&(F \alpha) ( s, a) := R(s,a) + \gamma \sum_{z \in Z} \max_{a' \in A} \sum_{s' \in S}  T^O(s, a, s', z) \alpha (s', a').
\end{split}
\end{equation}
It is straightforward to show that the FIB operator is a $\gamma$-contraction. 

\begin{lemma}\label{lem:contraction-fib}
The FIB operator $F$~\eqref{eq:fib} is a $\gamma$-contraction mapping with respect to $\| \cdot \|_\infty$, i.e.,
\[
\| F \alpha - F \alpha' \|_\infty  \leq \gamma \| \alpha - \alpha' \|_\infty \quad \alpha, \alpha' \in \mathbb{R}^{|S||A|}.
\]
\end{lemma}

\subsection{{{Soft}} FIB}
Here we present {\emph{soft FIB},} the maximum entropy regularization of FIB. 
As in the case of QMDP, the FIB update can be 
expressed as
\begin{equation}\label{eq:fib-phi}
\begin{split}
&\alpha^{k+1} (s, a) := R(s, a) + \gamma \sum_{z \in Z} \max_{\phi \in \mathcal{P}(A)} \sum_{a' \in A}  \phi(a') \tilde{\alpha}_{s,a,z}^k(a'),\\
\end{split}
\end{equation}
where $\tilde{\alpha}_{s,a,z}^k (a') := \sum_{s' \in S} T^O(s, a, s', z)  \alpha^k (s', a')$.
Then, the soft FIB update rule can be obtained as
\begin{equation*}
\begin{split}
&\alpha_a^{k+1} (s) := R(s, a) + \gamma \sum_{z \in Z} \ent^\star_{\tau, z}(\tilde{\alpha}_{s,a,z}^k), \\
&\ent^\star_{\tau,z}(\tilde{\alpha}_{s,a,z}^k) := \max_{\phi \in \mathcal{P}(A)} \left[ \sum_{a \in A} 
\phi(a) \tilde{\alpha}_{s,a,z}^k(a) + {\tau}\ent(\phi) \right] = \tau \ln \left ( \sum_{a'\in A} \exp
\frac{\tilde{\alpha}_{s,a,z}(a')}{\tau} \right ).
\end{split}
\end{equation*}
We define the soft FIB operator $F_\tau$ as follows:
\begin{equation}\label{eq:fib-maxent}
\begin{split}
&(F_\tau \alpha) ( s, a) := R(s,a) + \gamma \sum_{z\in Z} \tau \ln \left ( \sum_{a'\in A} \exp
\frac{\tilde{\alpha}_{s,a,z}(a')}{\tau} \right ).
\end{split}
\end{equation}
The soft FIB operator  shares most of the properties with the soft QMDP operator including the contraction property.

\section{Regularization with KL divergence}\label{appx:mellowmax}
Apart from the maximum entropy regularization, we can also consider an alternative regularization method using the Kullback-Leibler (KL) divergence between $\phi$ 
and the uniform distribution $\mathcal{U}_A \in \mathcal{P}(A)$, defined as
\begin{equation*}
\begin{split}
\kl(\phi\mid\mid\mathcal{U}_A) &= \sum_{a\in A} \phi(a) \ln\frac{\phi(a)}{\mathcal{U}_A(a)} \\
&= -\ent(\phi) + \ln |A|.
\end{split}
\end{equation*}
In this case, we augment the operator~\eqref{eq:qmdp-phi-vector} with $\tau\kl(\phi\mid\mid\mathcal{U}_A)$ as follows:
\begin{equation*}
\begin{split}
\alpha_a^{k+1} (s) &:= R(s, a) + \gamma \sum_{s' \in S} T (s, a, s') \mathcal{K}^\star_{\tau}(\alpha^k(s', \cdot)), \\
\mathcal{K}^\star_{\tau}(\alpha(s, \cdot)) 
&:= \max_{\phi \in \mathcal{P}(A)} \left[\sum_{a \in A} \phi(a) \alpha (s, a) - \tau\kl(\phi\mid\mid\mathcal{U}_A)  \right] \\
&= \ent^\star_{\tau}(\alpha(s, \cdot)) - \tau\ln |A|.
\end{split}
\end{equation*}
The corresponding fixed-point operator, which we call the \emph{kQMDP} operator, is defined as follows:
\begin{equation}\label{eq:qmdp-mellow}
\begin{split}
&(F_\tau \alpha) ( s, a) := R(s,a) + \gamma \sum_{s' \in S} T(s, a, s') \cdot \tau \ln \left (\frac{1}{|A|}  \sum_{a' \in A} \exp \frac{\alpha(s', a')}{\tau} \right ).
\end{split}
\end{equation}
Similarly, in the case of KL regularized FIB, the \emph{kFIB} operator  is defined as
\begin{equation}\label{eq:fib-mellow}
\begin{split}
&(F_\tau \alpha) ( s, a) := R(s,a) + \gamma \sum_{z\in Z} \tau \ln \left ( \frac{1}{|A|}\sum_{a'\in A}  \exp
\frac{\tilde{\alpha}_{s,a,z}(a')}{\tau} \right ).
\end{split}
\end{equation}
Since they differ from~\eqref{eq:qmdp-maxent} and~\eqref{eq:fib-maxent} by a constant, both~\eqref{eq:qmdp-mellow} and~\eqref{eq:fib-mellow} are also $\gamma$-contractions whose fixed points can be obtained via FPI. This type of regularization, also known as MellowMax~\cite{asadi2017alternative}, is numerically more desirable than the maximum entropy regularization especially when a large $\tau$ is used. By subtracting $\tau\ln|A|$ from $\mathcal{H}_{\tau}^\star$, we prevent $\mathcal{K}_\tau^\star$ from diverging as $\tau \to \infty$. Indeed, it is straightforward to derive the following inequalities for the KL divergence regularization:
\begin{equation}\label{eq:kl-div-ineq}
\frac{1}{|A|}\sum_{a} \alpha(s, a) \leq \mathcal{K}_\tau^\star(\alpha(s, \cdot)) \leq \max_{a} \alpha(s, a).
\end{equation}
In contrast to~\eqref{eq:entropy-ineq} where both the lower bound and the upper bound of $\mathcal{H}_{\tau}^\star$ depends linearly on $\tau$, the bounds appeared in~\eqref{eq:kl-div-ineq} are independent on $\tau$. This eliminates the possibility of performing arithmetic operations on extremely large numbers while this is unavoidable in the case of the maximum entropy regularization with large $\tau$.

\section{Results of Additional Experiments}\label{appx:exp}
In this section, we provide the results of  additional experiments using soft FIB and the KL divergence regularization.

\begin{table*}[tb] 
\caption{Experiment results:  Soft FIB ($\mathrm{mean} \pm \mathrm{std}$).}
\begin{center}
\scalebox{0.6}{
\begin{tabular}{|l|l|c|c|c|c|c|c|}
\hline
& & \multicolumn{6}{c|}{\textbf{Algorithm}}  \\  \hline
\multirow{2}{*}{\vtop{\hbox{\strut \textbf{Problem}}\hbox{\strut \textbf{($|\mathcal{S}|, |\mathcal{A}|, |\mathcal{O}|$)}}}} 
& \multirow{2}{*}{\textbf{Metric}} & \multirow{2}{*}{\textbf{FIB}} & \multicolumn{2}{c|}{\textbf{AA-FIB}}
& \multirow{2}{*}{\textbf{sFIB}} & \multicolumn{2}{c|}{\textbf{AA-sFIB}} \\
& & & without $\theta_\mathrm{targ}$ & with $\theta_\mathrm{targ}$ & & without $\theta_\mathrm{targ}$ & with $\theta_\mathrm{targ}$ \\\hline \hline


\multirow{5}{*}{\vtop{\hbox{\strut \emph{sunysb}}\hbox{\strut ($300,4,28$)}}}
& $\#\mathrm{iter}$ & 1361.06 $\pm$ 14.04 & 468.94 $\pm$ 173.61 & \textbf{119.99 $\pm$ 11.59} & 2427.00 $\pm$ 0.00 & 108.45 $\pm$ 2.13 & \textbf{101.88 $\pm$ 2.86} \\
& $\#\mathrm{AA}$  & N/A & 467.94 $\pm$ 173.61 & 25.35 $\pm$ 6.15 & N/A & 107.45 $\pm$ 2.13 & 70.64 $\pm$ 2.19 \\
& $t_{\mathrm{total}}$(sec) & 3.766 $\pm$ 0.036 & 1.629 $\pm$ 0.601 & \textbf{0.413 $\pm$ 0.040} & 26.552 $\pm$ 0.085 & 1.337 $\pm$ 0.027 & \textbf{1.238 $\pm$ 0.035} \\
& $\mathrm{reward}_{\mathrm{fixed}}$ & 0.760 $\pm$ 0.016 & 0.760 $\pm$ 0.016 & 0.759 $\pm$ 0.016 & 0.767 $\pm$ 0.015 & 0.765 $\pm$ 0.013 & 0.765 $\pm$ 0.015 \\
& $\mathrm{reward}_{\mathrm{rand}}$ & 0.387 $\pm$ 0.040 & 0.373 $\pm$ 0.041 & 0.390 $\pm$ 0.045 & 0.267 $\pm$ 0.038 & 0.264 $\pm$ 0.040 & 0.270 $\pm$ 0.043 \\ 
\hline

\multirow{5}{*}{\vtop{\hbox{\strut \emph{fourth}}\hbox{\strut ($1052,4,28$)}}}
& $\#\mathrm{iter}$ & 1360.43 $\pm$ 19.59 & 1203.26 $\pm$ 302.33 & \textbf{245.52 $\pm$ 30.55} & 2943.60 $\pm$ 2.57 & \textbf{112.29 $\pm$ 4.87} & 262.38 $\pm$ 65.98 \\
& $\#\mathrm{AA}$ & N/A & 1203.26 $\pm$ 302.33 & 24.03 $\pm$ 7.42 & N/A & 111.29 $\pm$ 4.87 & 61.36 $\pm$ 8.97 \\
& $t_{\mathrm{total}}$(sec) & 5.631 $\pm$ 0.085 & 6.160 $\pm$ 1.537 & \textbf{1.240 $\pm$ 0.156} & 82.376 $\pm$ 0.148 & \textbf{3.491 $\pm$ 0.152} & 8.153 $\pm$ 2.055 \\
& $\mathrm{reward}_{\mathrm{fixed}}$ & 0.596 $\pm$ 0.010 & 0.593 $\pm$ 0.011 & 0.594 $\pm$ 0.011 & 0.597 $\pm$ 0.011 & 0.513 $\pm$ 0.205 & 0.539 $\pm$ 0.170 \\
& $\mathrm{reward}_{\mathrm{rand}}$ & 0.368 $\pm$ 0.032 & 0.374 $\pm$ 0.034 & 0.371 $\pm$ 0.033 & 0.281 $\pm$ 0.033 & 0.277 $\pm$ 0.035 & 0.274 $\pm$ 0.038 \\   
\hline

\multirow{5}{*}{\vtop{\hbox{\strut \emph{tag}}\hbox{\strut ($870,5,30$)}}}
& $\#\mathrm{iter}$ & 315.61 $\pm$ 0.51 & 83.92 $\pm$ 6.02 & \textbf{96.06 $\pm$ 9.08} & 390.18 $\pm$ 0.38 & 93.86 $\pm$ 7.15 & \textbf{70.54 $\pm$ 2.27} \\
& $\#\mathrm{AA}$ & N/A & 82.92 $\pm$ 6.02 & 52.75 $\pm$ 7.17 & N/A & 92.86 $\pm$ 7.15 & 45.20 $\pm$ 1.89 \\
& $t_{\mathrm{total}}$(sec) & 1.533 $\pm$ 0.019 & 0.502 $\pm$ 0.037 & \textbf{0.564 $\pm$ 0.055} & 14.861 $\pm$ 0.029 & 3.944 $\pm$ 0.301 & \textbf{2.981 $\pm$ 0.098} \\
& $\mathrm{reward}_{\mathrm{fixed}}$ & $-17.258$ $\pm$ 0.685 & $-17.279$ $\pm$ 0.714 & $-17.268$ $\pm$ 0.680 & $-7.196$ $\pm$ 0.595 & $-7.147$ $\pm$ 0.577 & $-7.179$ $\pm$ 0.662 \\
& $\mathrm{reward}_{\mathrm{rand}}$ & $-15.867$ $\pm$ 0.839 & $-16.025$ $\pm$ 0.860 & $-15.970$ $\pm$ 0.887 & $-6.607$ $\pm$ 0.696 & $-6.632$ $\pm$ 0.653 & $-6.636$ $\pm$ 0.595 \\ 
\hline

\multirow{5}{*}{\vtop{\hbox{\strut \emph{underwater}}\hbox{\strut ($2653,6,102$)}}}
& $\#\mathrm{iter}$ & 446.24 $\pm$ 4.34 & 241.41 $\pm$ 16.57 & \textbf{79.03 $\pm$ 12.46} & 505.17 $\pm$ 0.38 & 225.45 $\pm$ 11.63 & \textbf{134.52 $\pm$ 12.93} \\
& $\#\mathrm{AA}$ & N/A & 240.41 $\pm$ 16.57 & 2.24 $\pm$ 0.76 & N/A & 224.45 $\pm$ 11.63 & 21.95 $\pm$ 1.92 \\
& $t_{\mathrm{total}}$(sec) & 15.398 $\pm$ 0.153 & 9.526 $\pm$ 0.655 & \textbf{3.074 $\pm$ 0.490} & 195.364 $\pm$ 0.586 & 90.840 $\pm$ 4.667 & \textbf{54.357 $\pm$ 5.244} \\
& $\mathrm{reward}_{\mathrm{fixed}}$ & $-51.776$ $\pm$ 46.712 & $-39.634$ $\pm$ 47.211 & $-46.334$ $\pm$ 49.199 & 681.926 $\pm$ 7.495 & 681.751 $\pm$ 9.405 & 680.746 $\pm$ 8.137 \\
& $\mathrm{reward}_{\mathrm{rand}}$ & 3235.844 $\pm$ 260.885 & 3305.371 $\pm$ 258.774 & 3263.062 $\pm$ 250.110 & 3254.562 $\pm$ 281.946 & 3258.112 $\pm$ 265.313 & 3270.237 $\pm$ 270.112 \\ 
\hline
\end{tabular}
}
\label{tab:aa-fib-maxent}
\end{center}
\end{table*}

The experimental results of soft FIB are presented in Table~\ref{tab:aa-fib-maxent}.
Note that AA-FIB without $\theta_\mathrm{targ}$ is equivalent to the AA-FIB algorithm in~\cite{ermis2021on}.
Overall, the number of iterations for (AA-)sFIB to converge to its fixed point is similar to that of (AA-)sQMDP,
but the computation time is much smaller in QMDP, as FIB has higher per-step computational complexity than QMDP.

Table~\ref{tab:aa-pomdp-mellow} presents the  results of our experiments for (AA-)kQMDP and (AA-)kFIB.
We first note that the fixed point of the kQMDP operator and that of the sQMDP operator only differs by a constant vector whose entries are identical, namely $\frac{\gamma\tau\ln|A|}{1-\gamma}\bm{1}$. Therefore, the policies obtained from these fixed points via~\eqref{eq:state-space-policy} are the same. Such a property is also observed in Table~\ref{tab:aa-pomdp-model} and Table~\ref{tab:aa-pomdp-mellow}, where (AA-)sQMDP and (AA-)kQMDP attain comparable cumulative rewards in all environments.
However, AA-kQMDP requires less per-iteration computational time compared to AA-sQMDP.
As a result, AA-kQMDP features lower $t_{\mathrm{total}}$ than AA-sQMDP.
A similar tendency is observed when comparing Table~\ref{tab:aa-fib-maxent} and Table~\ref{tab:aa-pomdp-mellow}, as $t_{\mathrm{total}}$ of  (AA-)kFIB is notably smaller than that of (AA-)sFIB.
For instance, using AA-sFIB in place of AA-kFIB in \emph{tag}, $t_{\mathrm{total}}$ decreases from 2.981sec to 0.725sec, which amounts to 75\% saving of the total computation time. In particular, AA-kFIB with $\theta_{\mathrm{targ}}$ (the last column of Table~\ref{tab:aa-pomdp-mellow}) spends only a half amount of $t_{\mathrm{total}}$ that is needed for the standard FIB (the first column of Table~\ref{tab:aa-fib-maxent}), while achieving higher $\mathrm{reward}_{\mathrm{fixed}}$.

\begin{table*}[t!] 
\caption{Experiment results: KL-regularized methods ($\mathrm{mean} \pm \mathrm{std}$).}
\begin{center}
\scalebox{0.6}{
\begin{tabular}{|l|l|c|c|c|c|c|c|}
\hline
& & \multicolumn{6}{c|}{\textbf{Algorithm}}  \\  \hline
\multirow{2}{*}{\vtop{\hbox{\strut \textbf{Problem}}\hbox{\strut \textbf{($|\mathcal{S}|, |\mathcal{A}|, |\mathcal{O}|$)}}}} 
& \multirow{2}{*}{\textbf{Metric}} & \multirow{2}{*}{\textbf{kQMDP}} & \multicolumn{2}{c|}{\textbf{AA-kQMDP}}
& \multirow{2}{*}{\textbf{kFIB}} & \multicolumn{2}{c|}{\textbf{AA-kFIB}} \\
& & & without $\theta_\mathrm{targ}$ & with $\theta_\mathrm{targ}$ & & without $\theta_\mathrm{targ}$ & with $\theta_\mathrm{targ}$ \\\hline \hline


\multirow{5}{*}{\vtop{\hbox{\strut \emph{sunysb}}\hbox{\strut ($300,4,28$)}}}
& $\#\mathrm{iter}$ & 1287.78 $\pm$ 31.00 & \textbf{86.23 $\pm$ 2.56} & 86.28 $\pm$ 2.50 & 1284.04 $\pm$ 32.39 & 86.73 $\pm$ 2.58 & \textbf{86.57 $\pm$ 2.65} \\
& $\#\mathrm{AA}$ & N/A & 85.23 $\pm$ 2.56 & 71.42 $\pm$ 2.41 & N/A & 85.73 $\pm$ 2.58 & 71.76 $\pm$ 2.65 \\
& $t_{\mathrm{total}}$(sec) & 0.267 $\pm$ 0.007 & \textbf{0.051 $\pm$ 0.002} & 0.051 $\pm$ 0.003 & 6.730 $\pm$ 0.173 & 0.527 $\pm$ 0.016 & \textbf{0.517 $\pm$ 0.016} \\
& $\mathrm{reward}_{\mathrm{fixed}}$ & 0.763 $\pm$ 0.015 & 0.762 $\pm$ 0.015 & 0.765 $\pm$ 0.015 & 0.765 $\pm$ 0.015 & 0.762 $\pm$ 0.015 & 0.764 $\pm$ 0.014 \\
& $\mathrm{reward}_{\mathrm{rand}}$ & 0.303 $\pm$ 0.039 & 0.300 $\pm$ 0.039 & 0.310 $\pm$ 0.050 & 0.263 $\pm$ 0.039 & 0.267 $\pm$ 0.046 & 0.269 $\pm$ 0.046 \\ 
\hline

\multirow{5}{*}{\vtop{\hbox{\strut \emph{fourth}}\hbox{\strut ($1052,4,28$)}}}
& $\#\mathrm{iter}$ & 1363.33 $\pm$ 1.19 & 96.48 $\pm$ 1.82 & \textbf{96.19 $\pm$ 6.10} & 1285.90 $\pm$ 27.11 & \textbf{89.02 $\pm$ 1.86} & 92.07 $\pm$ 2.20 \\
& $\#\mathrm{AA}$ & N/A & 95.48 $\pm$ 1.82 & 74.63 $\pm$ 6.76 & N/A & 88.02 $\pm$ 1.86 & 72.65 $\pm$ 2.17 \\
& $t_{\mathrm{total}}$(sec) & 0.449 $\pm$ 0.003 & 0.093 $\pm$ 0.003 & \textbf{0.092 $\pm$ 0.006} & 11.070 $\pm$ 0.249 & \textbf{0.884 $\pm$ 0.019} & 0.913 $\pm$ 0.023 \\
& $\mathrm{reward}_{\mathrm{fixed}}$ & 0.602 $\pm$ 0.008 & 0.595 $\pm$ 0.060 & 0.599 $\pm$ 0.008 & 0.597 $\pm$ 0.009 & 0.493 $\pm$ 0.221 & 0.551 $\pm$ 0.151 \\
& $\mathrm{reward}_{\mathrm{rand}}$ & 0.323 $\pm$ 0.037 & 0.309 $\pm$ 0.045 & 0.305 $\pm$ 0.033 & 0.285 $\pm$ 0.035 & 0.274 $\pm$ 0.039 & 0.280 $\pm$ 0.037 \\ 
\hline

\multirow{5}{*}{\vtop{\hbox{\strut \emph{tag}}\hbox{\strut ($870,5,30$)}}}
& $\#\mathrm{iter}$ & 304.98 $\pm$ 3.06 & 58.03 $\pm$ 1.21 & \textbf{57.93 $\pm$ 1.12} & 303.97 $\pm$ 3.06 & 57.81 $\pm$ 1.40 & \textbf{57.77 $\pm$ 1.26} \\
& $\#\mathrm{AA}$ & N/A & 57.03 $\pm$ 1.21 & 40.96 $\pm$ 1.10 & N/A & 56.81 $\pm$ 1.40 & 40.98 $\pm$ 1.16 \\
& $t_{\mathrm{total}}$(sec) & 0.131 $\pm$ 0.002 & 0.068 $\pm$ 0.003 & \textbf{0.068 $\pm$ 0.003} & 3.331 $\pm$ 0.037 & 0.735 $\pm$ 0.018 & \textbf{0.725 $\pm$ 0.016} \\
& $\mathrm{reward}_{\mathrm{fixed}}$ & $-6.609$ $\pm$ 0.568 & $-6.757$ $\pm$ 0.666 & $-6.835$ $\pm$ 0.566 & $-6.571$ $\pm$ 0.645 & $-6.684$ $\pm$ 0.620 & $-6.717$ $\pm$ 0.649 \\
& $\mathrm{reward}_{\mathrm{rand}}$ & $-6.369$ $\pm$ 0.588 & $-6.405$ $\pm$ 0.592 & $-6.377$ $\pm$ 0.600 & $-6.455$ $\pm$ 0.710 & $-6.376$ $\pm$ 0.647 & $-6.293$ $\pm$ 0.558 \\ 
\hline

\multirow{5}{*}{\vtop{\hbox{\strut \emph{underwater}}\hbox{\strut ($2653,6,102$)}}}
& $\#\mathrm{iter}$ & 435.86 $\pm$ 7.97 & 159.77 $\pm$ 6.28 & \textbf{122.25 $\pm$ 2.29} & 446.36 $\pm$ 4.09 & 171.23 $\pm$ 7.63 & \textbf{133.02 $\pm$ 12.41} \\
& $\#\mathrm{AA}$ & N/A & 158.77 $\pm$ 6.28 & 18.91 $\pm$ 1.48 & N/A & 170.23 $\pm$ 7.63 & 22.28 $\pm$ 1.97 \\
& $t_{\mathrm{total}}$(sec) & 0.486 $\pm$ 0.009 & 0.543 $\pm$ 0.021 & \textbf{0.416 $\pm$ 0.011} & 44.588 $\pm$ 0.431 & 18.713 $\pm$ 0.830 & \textbf{14.723 $\pm$ 1.388} \\
& $\mathrm{reward}_{\mathrm{fixed}}$ & 685.038 $\pm$ 8.047 & 683.502 $\pm$ 7.883 & 684.241 $\pm$ 8.564 & 680.497 $\pm$ 9.450 & 680.785 $\pm$ 8.376 & 680.551 $\pm$ 9.169 \\
& $\mathrm{reward}_{\mathrm{rand}}$ & 3245.707 $\pm$ 230.666 & 3225.733 $\pm$ 257.239 & 3276.819 $\pm$ 246.870  & 3214.878 $\pm$ 254.658 & 3219.718 $\pm$ 279.785 & 3258.293 $\pm$ 276.259 \\ 
\hline
\end{tabular}
}
\label{tab:aa-pomdp-mellow}
\end{center}
\end{table*}

\bibliographystyle{IEEEtran}
\bibliography{reference}

\end{document}